\documentclass[11pt,a4paper]{article}

\usepackage[utf8]{inputenc}
\usepackage[bookmarks]{hyperref}
\hypersetup{colorlinks=true,citecolor=blue,linkcolor=blue,filecolor=blue,urlcolor=blue,linktocpage,hyperfootnotes=false}
\usepackage{amsmath,amssymb,amsthm}
\numberwithin{equation}{section}
\usepackage{dsfont}
\usepackage{fullpage}
\usepackage[affil-it]{authblk}

\usepackage{microtype}

\usepackage{cleveref}

\usepackage{enumerate} 
\usepackage{mathtools}

\allowdisplaybreaks[4]

\newtheorem{theorem}{Theorem}[section]
\newtheorem{lemma}[theorem]{Lemma}
\newtheorem{proposition}[theorem]{Proposition}

\theoremstyle{definition}
\newtheorem{definition}[theorem]{Definition}
\newtheorem{remark}[theorem]{Remark}

\newcommand\numberthis{\addtocounter{equation}{1}\tag{\theequation}}

\usepackage{mathpazo}

\usepackage{tikz}
\usepackage{pgfplots}
\usetikzlibrary{arrows,positioning}
\usetikzlibrary{decorations.pathmorphing}
\usetikzlibrary{patterns}
\tikzset{snake it/.style={decorate, decoration=snake}}
\pgfdeclarepatternformonly{north east lines wide}{\pgfqpoint{-1pt}{-1pt}} {\pgfqpoint{10pt}{10pt}} {\pgfqpoint{9pt}{9pt}} 
{
     \pgfsetlinewidth{0.4pt}
     \pgfpathmoveto{\pgfqpoint{0pt}{0pt}}
     \pgfpathlineto{\pgfqpoint{9.1pt}{9.1pt}}
     \pgfusepath{stroke}
    }

\definecolor{lightgreen}{RGB}{163, 236, 163}
\definecolor{lightred}{RGB}{255, 153, 153}
\definecolor{darkred}{RGB}{178, 0, 0}
\definecolor{us}{RGB}{0, 102, 255}
\definecolor{lightus}{RGB}{128, 178, 255}

\tikzset{>=stealth', pil/.style={->,thick,shorten <=2pt,shorten >=2pt,}}

\usepackage{graphicx}
\usepackage[font=small]{caption}
\usepackage{subcaption}

\newcommand{\cB}{\mathcal{B}}
\newcommand{\cC}{\mathcal{C}}
\newcommand{\cD}{\mathcal{D}}
\newcommand{\cE}{\mathcal{E}}
\newcommand{\cH}{\mathcal{H}}
\newcommand{\cP}{\mathcal{P}}
\newcommand{\cX}{\mathcal{X}}
\newcommand{\cZ}{\mathcal{Z}}

\newcommand{\tD}{\widetilde{D}}
\newcommand{\tS}{\widetilde{S}}
\newcommand{\tI}{\widetilde{I}}

\newcommand{\hX}{\hat{X}}

\newcommand{\bX}{\bar{X}}

\DeclareMathOperator{\tr}{Tr}
\DeclareMathOperator{\id}{id}
\DeclareMathOperator{\supp}{supp}
\DeclareMathOperator{\rk}{rk}

\newcommand{\eps}{\varepsilon}
\newcommand{\sumi}{\sum\nolimits}
\newcommand{\onehalf}{\frac{1}{2}}
\newcommand{\one}{I}
\newcommand{\ox}{\otimes}

\newcommand{\Qncsm}{q_{\text{\normalfont csm}}}
\newcommand{\Encsm}{e_{\text{\normalfont csm}}}
\newcommand{\Qnqss}{q_{\text{\normalfont qss}}}
\newcommand{\Enqss}{e_{\text{\normalfont qss}}}

\newcommand{\qfw}{q_{\rightarrow}}
\newcommand{\qtot}{q_{\leftrightarrow}}

\newcommand{\vphi}{\varphi}
\newcommand{\psuc}{p_{\text{\normalfont succ}}}
\newcommand{\bomega}{\bar{\omega}}
\newcommand{\bsigma}{\bar{\sigma}}

\newcommand{\qqquad}{\hspace{2cm}}

\newcommand{\etal}{\textit{et al}.~}

\usepackage[backend=bibtex8,style=ieee,sorting=nty,sortcites=false,firstinits=true,doi=false,isbn=false,url=false,maxbibnames=5]{biblatex}
\renewbibmacro{in:}{}
\begin{document}
\title{\textbf{Strong converse theorems using Rényi entropies}}

\author[a]{Felix Leditzky}
\author[b]{Mark M.~Wilde}
\author[a]{Nilanjana Datta}
\affil[a]{\small Statistical Laboratory, Centre for Mathematical Sciences, University of Cambridge, Cambridge CB3 0WB, UK}
\affil[b]{\small Hearne Institute for Theoretical Physics, Department of Physics and Astronomy, Center for Computation and Technology, Louisiana State University, Baton Rouge, Louisiana 70803, USA}
\maketitle

\begin{abstract}
We use a Rényi entropy method to prove strong converse theorems for certain infor\-mation-theoretic tasks which involve local operations and quantum or classical communication between two parties. 
These include state redistribution, coherent state merging, quantum state splitting, measurement compression with quantum side information, randomness extraction against quantum side information, and data compression with quantum side information. 
The method we employ in proving these results extends ideas developed by Sharma \cite{Sha14}, which he used to give a new proof of the strong converse theorem for state merging.
For state redistribution, we prove the strong converse property for the boundary of the entire achievable rate region in the $(e,q)$-plane, where $e$ and $q$ denote the entanglement cost and quantum communication cost, respectively.
In the case of measurement compression with quantum side information, we prove a strong converse theorem for the classical communication cost, which is a new result extending the previously known weak converse.
For the remaining tasks, we provide new proofs for strong converse theorems previously established using smooth entropies. 
For each task, we obtain the strong converse theorem from explicit bounds on the figure of merit of the task in terms of a Rényi generalization of the optimal rate. Hence, we identify candidates for the strong converse exponents for each task discussed in this paper.
To prove our results, we establish various new entropic inequalities, which might be of independent interest. 
These involve conditional entropies and mutual information derived from the sandwiched R\'enyi divergence. 
In particular, we obtain novel bounds relating these quantities, as well as the R\'enyi conditional mutual information, to the fidelity of two quantum states. 
\end{abstract}

\section{Introduction}
\subsection{Strong converse theorems and the Rényi entropy method}\label{sec:strong-converse-property}
One of the primary goals of quantum information theory is to find optimal rates of information-theoretic tasks, such as data compression, information transmission through a noisy quantum channel, and entanglement manipulation. 
Depending on the specific task in question, the optimal rate is either an optimal {\em{cost}}, quantifying the minimum rate at which an available resource is consumed in accomplishing the task, or an optimal {\em{gain}}, quantifying the maximum rate at which a desired target resource is produced in the process. 
For any rate above (below) the optimal cost (gain) there is a corresponding protocol for accomplishing the task successfully. 
That is, the error, $\eps_n$, incurred in the protocol for $n$ uses of the underlying resource vanishes in the asymptotic limit ($n \to \infty$).\footnote{For example, in the case of information transmission through a memoryless noisy channel, $n$ denotes the number of independent uses of the channel.}
Such rates are said to be {\em{achievable}}. 
In contrast, for protocols with non-achievable rates, the error does not vanish asymptotically.
The optimal rate of an information-theoretic task is said to satisfy the {\em{strong converse property}} if for any sequence of protocols with a non-achievable rate the error $\eps_n$ is not only bounded away from zero but necessarily converges to one in the asymptotic limit. 
Moreover, this convergence is exponential\footnote{The requirement of exponential convergence is sometimes relaxed when defining the strong converse property.} in $n$, that is,
\begin{align}\label{eq:strong-converse}
\eps_n \geq 1 - \exp(-Kn),
\end{align} 
for some positive constant $K$. 
A {\em{strong converse theorem}} establishes the strong converse property of the optimal rate of an information-theoretic task, and hence serves to identify the latter as a sharp rate threshold for the task.

For information transmission through classical noisy channels, the strong converse theorem was first proved by Wolfowitz \cite{Wol61}. 
An alternate proof of this theorem was later given by Arimoto \cite{Ari73} by employing the properties of a quantity which is sometimes referred to as the Gallager function \cite{PV10}.
Ogawa and Nagaoka \cite{ON99} extended this method to the quantum setting to prove the strong converse property of the capacity of a classical-quantum channel, which was also proved concurrently  by Winter \cite{Win99} using the method of types. 
Nagaoka \cite{Nag01} further developed Arimoto's idea to give a new proof of this result.
To this end, he employed a Rényi divergence and its monotonicity under completely positive, trace-preserving maps (also called the {\em{data processing inequality}}), establishing what we refer to as the `Rényi entropy method' in this paper. 
Later, Polyanskiy and Verdú \cite{PV10} realized that it is possible to establish converse bounds by employing any divergence satisfying the data processing inequality.
In \cite{FT14} the Rényi entropy method was used to obtain strong converse theorems in network information theory.

Different quantum generalizations of the $\alpha$-Rényi divergence have been introduced \cite{Pet86,MDSFT13,WWY14} and their monotonicity under quantum operations for certain ranges of the Rényi parameter $\alpha$ has been established \cite{Pet86,FL13,Bei13}. 
The Rényi entropy method has since been successfully employed to prove strong converse theorems for classical channel coding with entangled inputs for a large class of quantum channels with additive Holevo capacity \cite{KW09}. 
More recently, strong converse theorems were proved for classical information transmission through entanglement-breaking and Hadamard channels \cite{WWY14} and quantum information transmission through generalized dephasing channels \cite{TWW14}. 
For the task of quantum hypothesis testing, Mosonyi and Ogawa \cite{MO13} established the `sandwiched Rényi divergence' of order $\alpha$ \cite{MDSFT13,WWY14} as a strong converse exponent. 
This was generalized by Cooney \etal \cite{CMW14}, who established the sandwiched Rényi divergence as a strong converse exponent in adaptive channel discrimination for certain channels.
Hayashi and Tomamichel \cite{HT14} showed that, in binary quantum hypothesis testing, for a composite alternative hypothesis the strong converse exponents are given by a Rényi mutual information and Rényi conditional entropy defined in terms of the sandwiched Rényi divergence. 

Application of the Rényi entropy method to prove the strong converse property for an information-theoretic task involving local operations and classical communication (LOCC) between two parties (say, Alice and Bob) was considered by Hayashi \etal \cite{HKMMW03} in the context of entanglement concentration (see also \cite{Hay06a}). 
More recently, Sharma \cite{Sha14} used the Rényi entropy method to establish the strong converse theorem for the task of {\em{state merging}}: 
Alice and Bob initially share a bipartite state and the aim is for Alice to transfer her part of the state to Bob by sending information to him through a noiseless classical channel. 
Both Alice and Bob are also allowed to  make use of prior shared entanglement between them, to assist them in achieving this task. 
In this case monotonicity of a relevant Rényi divergence under LOCC plays a pivotal role in establishing the strong converse for the optimal entanglement cost.\footnote{This result also follows from prior work by various authors \cite{Win14,Ber09,Tom12}.}

In this paper, we further extend the Rényi entropy method to establish strong converse theorems for a range of information-theoretic tasks which involve local operations and quantum or classical communication (LOQC or LOCC) between two parties. 
These tasks (which are elucidated at the beginning of each section respectively) include state redistribution (with and without feedback), coherent state merging, quantum state splitting, measurement compression with quantum side information, randomness extraction against quantum side information, and data compression with quantum side information. 
Strong converse theorems for some of these tasks were previously obtained in the so-called smooth entropy framework introduced by Renner \cite{Ren05}. 
This was done for the quantum communication cost of state redistribution by Berta \textit{et al.}~\cite{BCT14v2}, and for randomness extraction against quantum side information and data compression with quantum side information by Tomamichel \cite{Tom12}.

Two inequalities which we use in proving the strong converse theorems are due to van Dam and Hayden \cite{vDH02}: the first bounds the fidelity between two states in terms of the difference of their Rényi entropies, and the second is a certain subadditivity property of the $\alpha$-Rényi entropy of a bipartite state. In addition, we establish various new inequalities involving conditional entropies and mutual information derived from the sandwiched R\'enyi divergence of order $\alpha$ \cite{MDSFT13,WWY14}. These inequalities play a crucial role in the proofs of our strong converse theorems and might also be of independent interest.   

Let us use the example of quantum data compression to explain the key step of the Rényi entropy method which establishes the strong converse property (cf.~\cite{Hay02}).
Let $\rho$ denote the source state of a discrete, memoryless source.
In this case, the optimal rate $r^*$ is the minimum rate of data compression and is given \cite{Sch95} by the von Neumann entropy of the source: $r^*= S(\rho)\coloneqq -\tr(\rho\log\rho)$.\footnote{In this paper all logarithms and exponentials are taken to base $2$.}
Consider a data compression protocol of rate $r$. The key step of the Rényi entropy method applied to quantum data compression is to prove that for values of a real parameter $\alpha>1$, there exists a positive constant $\kappa(\alpha)$ such that the error, $\eps_n$, incurred for $n$ independent uses of the source, satisfies the following bound \cite{Sha14}:
\begin{align}\label{eq:renyi-approach-bound}
\eps_n \geq 1 - \exp\left\lbrace -n \kappa(\alpha) \left[S_{\alpha}(\rho)-r\right] \right\rbrace.
\end{align}
Here, $\kappa(\alpha)=(\alpha-1)/(2\alpha)$, and $S_{\alpha}(\rho)$ is the Rényi entropy of the source state $\rho$:
\begin{align*}
S_{\alpha}(\rho) \coloneqq \frac{1}{1-\alpha} \log \tr \rho^\alpha,
\end{align*}
satisfying (cf.~\Cref{prop:renyi-properties})
\begin{align}\label{eq:renyi-convergence}
\begin{aligned}
S_{\alpha}(\rho) \leq S_{\alpha'}(\rho) \leq S(\rho)\quad\text{for }\alpha\geq\alpha'> 1 \qquad\text{and}\qquad \lim\nolimits_{\alpha\rightarrow 1} S_\alpha(\rho) = S(\rho)=r^*.
\end{aligned}
\end{align}
The inequality \eqref{eq:renyi-approach-bound}, along with the statements in \eqref{eq:renyi-convergence}, readily imply the statement \eqref{eq:strong-converse} of the strong converse theorem. 
To see this, suppose $r < r^*=S(\rho)$. 
Then there is a $\delta>0$ such that $r+\delta < r^*$. 
Moreover, \eqref{eq:renyi-convergence} implies that for every $\delta>0$ there is an $\alpha_0>1$ such that $S_{\alpha_0}(\rho) > r^* - \delta$. 
Combining the two inequalities, we get $r<r^*-\delta < S_{\alpha_0}(\rho)$, and inserting this in \eqref{eq:renyi-approach-bound} yields the strong converse condition as stated in \eqref{eq:strong-converse} with the choice $K \coloneqq \kappa(\alpha_0)[S_{\alpha_0}(\rho)-r] > 0$.

\subsection{Main results and structure of the paper}
In the present paper, we prove strong converse theorems for the following tasks using the Rényi entropy method:
\begin{enumerate}[(a)]
\item\label{item:state-re} state redistribution (with and without quantum feedback), coherent state merging and quantum state splitting;
\item measurement compression with quantum side information (QSI);
\item\label{item:randex} randomness extraction against QSI;
\item\label{item:dc} data compression with QSI.
\end{enumerate}
Previously, Tomamichel \cite{Tom12} proved strong converse theorems for randomness extraction against QSI and data compression with QSI using the {\em{smooth entropy framework}}. 
Recently, Berta \etal \cite{BCT14v2} proved a strong converse theorem for the quantum communication cost in state redistribution (which holds even in the presence of quantum feedback). 
However, their strong converse theorem did not establish the strong converse property for the boundary of the entire achievable rate region in the $(e,q)$-plane, where $e$ and $q$ denote the entanglement cost and quantum communication cost, respectively (see \Cref{fig:achievable-rate-region}). 
In this paper, we fill this gap with \Cref{thm:state-re-strong-converse} (for the case of no feedback) and \Cref{thm:state-re-fb-strong-converse} (for the case of feedback). 
The study of the strong converse for state redistribution with feedback was inspired by \cite{BCT14v2}, where this protocol was first defined. 
Note that, following discussions with the authors of the present paper, Berta \etal have now also obtained a strong converse theorem for the entire achievable region \cite{BCT14} for state redistribution with and without feedback.
In the case of measurement compression with quantum side information, our strong converse theorem for the classical communication cost is a new result, which strengthens the previously established weak converse of \cite{WHBH12}.
We also provide alternative proofs for the strong converse theorems of the protocols in \cref{item:randex,item:dc} in the above list using the Rényi entropy approach. 

\begin{figure}[t]
\centering
\begin{tikzpicture}[scale=1]
    \tikzstyle{every node}=[font=\large];
	\begin{scope}
	\clip (-1.5,5.875) -- (6,5.875)-- (6,1.5) -- (2,1.5) -- (-1.5,5.875);
	\fill[color=lightgreen] (-1.5,6) rectangle (6,1.5);
	\end{scope}
	
	\begin{scope}
	\clip (-1.5,0) -- (-1.5,5.875) -- (2,1.5) -- (6,1.5) -- (6,0) -- (-1.5,0);
	\fill[pattern = north east lines wide, pattern color=darkred] (-1.5,0) rectangle (6,5.875);
	\end{scope}
		
	\draw[ultra thick,domain=-1.5:2] plot (\x,4-1.25*\x);
	\draw[ultra thick,domain=-2:-1.5,dashed] plot (\x,4-1.25*\x);
	\draw[ultra thick,domain=6:7,dashed] plot (\x,1.5);
	\draw[ultra thick,domain=2:6] plot (\x,1.5);
	\draw[thick, dashed,domain=-2:2] plot(\x,1.5) ;
	
	\node at (-3.5,6) {$q+e\geq S(A|B)_\rho$};
	\node[align=center] at (3.5,3.5) {achievable\\ region};
	\node at (-3.5,1.5) {$\onehalf I(A;R|B)_\rho$};
	\node[align = center] at (8.0,0.75) {strong converse\\ region};
	\node[align = center] at (-3.5,3.5) {strong converse\\ region};

	\draw[thick,->] (-2,0)--(7,0) node[below=1ex,pos=0.99]{$e$};
	\draw[thick,->] (0,-0.2)--(0,6.5) node[left=1ex,pos=0.99]{$q$};
	\draw[thick,dashed] (2,1.5)--(2,-0.5) node[below] {$\onehalf I(A;C)_\rho-\onehalf I(A;B)_\rho$};
\end{tikzpicture}
\caption{Plot of the plane of rate pairs $(e,q)$ for state redistribution, where $e$ is the entanglement cost \eqref{eq:state-re-ent-cost} and $q$ is the quantum communication cost \eqref{eq:state-re-qu-comm-cost}. 
The shaded area is the region of achievable rate pairs defined by $\lbrace (e,q)\colon q + e \geq S(A|B)_\rho\text{ and } q \geq \onehalf I(A;R|B)_\rho\rbrace$. 
The hatched area is the strong converse region, as proved in the present paper (\Cref{thm:state-re-strong-converse}) and by Berta \etal \cite{BCT14}. 
For a definition of the state redistribution protocol, see \Cref{sec:state-re-protocol}.}
\label{fig:achievable-rate-region}
\end{figure}
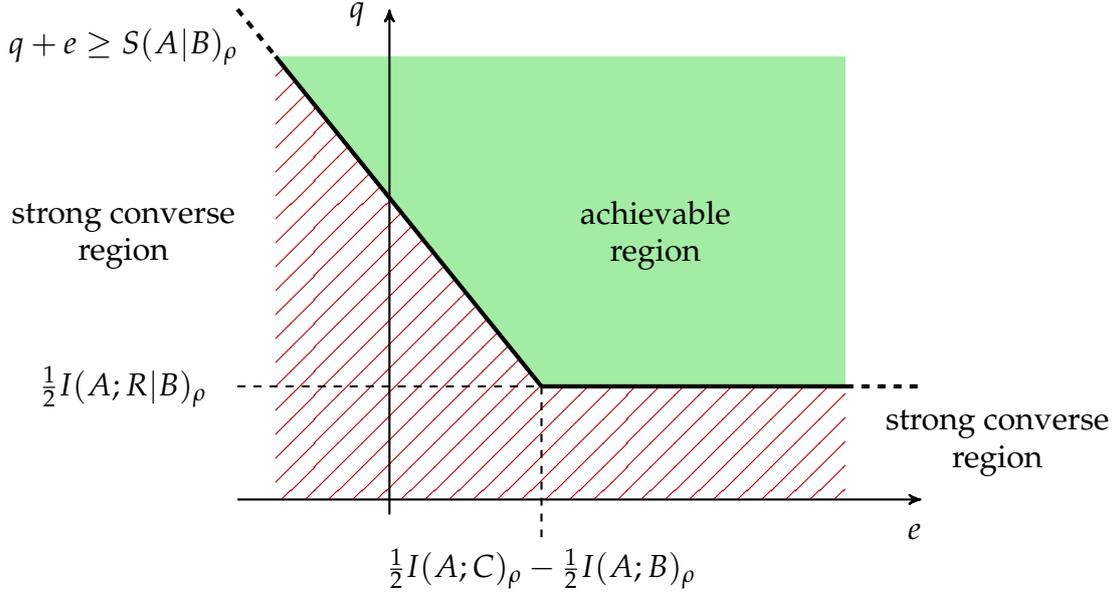

Our strong converse theorems are given in terms of various R\'enyi generalizations (see \Cref{subsec:renyi-entropies}) of the optimal rates of the protocols, whose properties we derive in \Cref{sec:renyi-properties}. 
In particular, we establish various new inequalities involving conditional entropies and mutual information derived from the sandwiched R\'enyi divergence. 
These include novel bounds relating these entropic quantities, as well as the R\'enyi conditional mutual information (defined in \ref{eq:renyi-cmi}), to the fidelity of two quantum states (see \Cref{prop:renyi-fidelity}). These bounds play a key role in our proofs.

This paper is structured as follows. 
In \Cref{sec:preliminaries} we set the notation, introduce definitions of the required R\'enyi entropic quantities, and discuss their properties. 
In the following sections we subsequently prove strong converse theorems for the protocols of state redistribution (\Cref{sec:state-re}), measurement compression with quantum side information (\Cref{sec:measurement-comp}), randomness extraction against quantum side information (\Cref{sec:randex}), and data compression with quantum side information (\Cref{sec:dc}). 
In the case of state redistribution, we elucidate the fact that it serves as a generalization of coherent state merging (\Cref{sec:fqsw-protocol}) and quantum state splitting (\Cref{sec:state-splitting}), proving strong converse theorems for the latter protocols as well. 
We summarize our results and discuss open questions in \Cref{sec:discussion}.
\Cref{sec:renyi-properties-proofs,sec:feedback-proofs} contain some of the proofs of \Cref{sec:renyi-properties,sec:feedback}.

\section{Preliminaries}\label{sec:preliminaries}
\subsection{Notation \& definitions}
For a finite-dimensional Hilbert space $\cH$ we denote the set of linear operators acting on $\cH$ by $\cB(\cH)$. 
We define the set of positive semi-definite operators $\cP(\cH)\coloneqq \lbrace X\in\cB(\cH)\colon X\geq 0\rbrace$ and refer to $P\in\cP(\cH)$ simply as a positive operator. 
The set $\cD(\cH)$ of density operators (or quantum states) on $\cH$ is the set of positive operators with unit trace: $\cD(\cH)\coloneqq \lbrace \rho\in\cP(\cH)\colon \tr\rho=1\rbrace$. 
The support $\supp Q$ of an operator $Q$ is defined as the orthogonal complement of its kernel. 
We write $A\not\perp B$ if $\supp A\cap \supp B$ contains at least one non-zero vector.
For a quantum system $A$ with associated Hilbert space $\cH_A$ we write $|A|\coloneqq \dim\cH_A$. 
If $B$ is another quantum system with associated Hilbert space $\cH_B$, then we set $\cH_{AB}\coloneqq \cH_A\ox\cH_B$. 
We write $A\cong B$ for quantum systems $A$ and $B$ whose underlying Hilbert spaces are isomorphic.

For a pure state $|\psi_A\rangle\in\cH_A$ we make the abbreviation $\psi_A\equiv |\psi\rangle\langle\psi|_A\in\cD(\cH_A)$. 
We denote by $\one_A$ the identity operator acting on $\cH_A$, and by $\id_A$ the identity superoperator acting on $\cD(\cH_A)$. 
The completely mixed state on $\cH_A$ is denoted by $\pi_A\coloneqq |A|^{-1}\one_A$. 
Let $\cX$ be some finite set and $\cH_X$ the associated Hilbert space with orthonormal basis $\lbrace |x\rangle\rbrace_{x\in\cX}$.
Then the quantum embedding of a classical register $X$ corresponding to $\cX$ is defined as the subspace of diagonal matrices in $\cD(\cH_X)$. 
A classical-quantum (c-q) state $\rho_{XB}\in\cD(\cH_{XB})$ is defined as $\rho_{XB} \coloneqq \sum_{x\in\cX} p_x |x\rangle\langle x|_X \ox \rho_B^x$ where $\rho_B^x\in\cD(\cH_B)$ for all $x\in\cX$. 
For quantum systems $A\cong B$, a maximally entangled state (MES) $|\Phi_{AB}\rangle$ is defined as 
\begin{align*}
|\Phi_{AB}\rangle\coloneqq \frac{1}{\sqrt{|A|}}\sum_{i=1}^{|A|} |i_A\rangle\ox|i_B\rangle,
\end{align*} 
where $\lbrace |i_A\rangle\rbrace_{i=1}^{|A|}$ and $\lbrace |i_B\rangle\rbrace_{i=1}^{|B|}$ are orthonormal bases for $\cH_A$ and $\cH_B$, respectively. 
We also use the notation $\Phi_{AB}^k$ to explicitly indicate the Schmidt rank $k=|A|=|B|$ of the MES.

A quantum operation $\Lambda$ is a linear, completely positive, trace-preserving (CPTP) map. 
For a quantum operation $\Lambda\colon \cD(\cH_A)\rightarrow \cD(\cH_B)$ we write $\Lambda\colon A\rightarrow B$. 
For quantum states $\rho,\sigma\in\cD(\cH)$, we define the \emph{fidelity} $F(\rho,\sigma)$ between $\rho$ and $\sigma$ as
\begin{align*}
F(\rho,\sigma) \coloneqq \left\|\sqrt{\rho}\sqrt{\sigma} \right\|_1.
\end{align*}
The \emph{von Neumann entropy} of a state $\rho_A\in\cD(\cH_A)$ is given by $S(\rho_A) \coloneqq -\tr(\rho_A\log\rho_A)$, and we use the notation $S(A)_\rho\equiv S(\rho_A)$. 
The \emph{quantum relative entropy} of a state $\rho\in\cD(\cH)$ and a positive operator $\sigma\in\cP(\cH)$ is defined by 
\begin{align}\label{eq:quantum-relative-entropy}
D(\rho\|\sigma)\coloneqq \tr\left[\rho(\log\rho-\log\sigma)\right]
\end{align}
if $\supp\rho\subseteq\supp\sigma$, and set to be equal to $+\infty$ otherwise.
Furthermore, we consider the following quantities for a tripartite state $\rho_{ABC}\in\cD(\cH_{ABC})$ and its marginals:
\begin{itemize}
\item the \emph{quantum conditional entropy} $S(A|B)_\rho \coloneqq S(AB)_\rho - S(B)_\rho$

\item the \emph{quantum mutual information} $I(A;B)_\rho \coloneqq S(A)_\rho + S(B)_\rho - S(AB)_\rho$

\item the \emph{conditional quantum mutual information} $I(A;B|C)_\rho \coloneqq S(A|C)_\rho + S(B|C)_\rho - S(AB|C)_\rho$
\end{itemize}

\subsection{Schatten norms}
\begin{definition}[Schatten $p$-(quasi)norm]
Let $M\in\cB(\cH)$ and $p>0$. 
Then we define
\begin{align*}
\|M\|_p \coloneqq \left(\tr |M|^p\right)^{1/p},
\end{align*}
where $|M|\coloneqq \sqrt{M^\dagger M}$. 
For $1\leq p\leq\infty$, the functional $\|\cdot\|_p$ defines a norm, the \emph{Schatten $p$-norm}.
\end{definition}

\begin{theorem}\label{thm:schatten}
Let $M,N\in\cB(\cH)$.
\begin{enumerate}[{\normalfont (i)}]
\item\label{item:schatten-hoelders-ineq} H\"older's inequality: 
Let $1\leq p \leq \infty$ and $q$ be the H\"older conjugate of $p$ defined by $\frac{1}{p} + \frac{1}{q} = 1$.
Then
\begin{align}\label{eq:hoelder-inequality}
\|MN\|_1\leq \|M\|_p\|N\|_q.
\end{align}

\item\label{item:schatten-mccarthy-ineq} McCarthy's inequalities {\normalfont \cite{McC67}}: 
For $p\in(0,1)$ we have
\begin{align}\label{eq:mccarthy-quasinorm}
\|M+N\|_p^p \leq \|M\|_p^p + \|N\|_p^p,
\end{align}
whereas for $1\leq p\leq \infty$ we have
\begin{align}\label{eq:mccarthy-norm}
\|M\|_p^p + \|N\|_p^p \leq \|M+N\|_p^p.
\end{align}
\end{enumerate}
\end{theorem}

\begin{remark}
For $1\leq p\leq\infty$, the functional $\|\cdot\|_p$ satisfies the triangle inequality, which for the Schatten $p$-norms is also known as the \emph{Minkowski inequality}:
\begin{align}\label{eq:minkowski}
\|M+N\|_p \leq \|M\|_p + \|N\|_p.
\end{align}
Hence, $\|\cdot\|_p$ defines a norm for this range of $p$. 
However, for $p\in(0,1)$ the Minkowski inequality \eqref{eq:minkowski} fails to hold, and we have the weaker inequality \eqref{eq:mccarthy-quasinorm} instead. 
Therefore, $\|\cdot\|_p$ only defines a quasinorm for the range $p\in(0,1)$.
\end{remark}

\subsection{Rényi entropies}\label{subsec:renyi-entropies}
\begin{definition}[\cite{MDSFT13,WWY14}]\label{def:renyi-entropies}~
\begin{enumerate}[{\normalfont (i)}]
\item Let $\alpha\in (0,\infty)\setminus \lbrace 1\rbrace$, and $\rho,\sigma\in\cP(\cH)$. If $\supp\rho\subseteq\supp\sigma$ for $\alpha>1$ or $\rho\not\perp\sigma$ for $\alpha\in(0,1)$, the \emph{sandwiched Rényi divergence of order $\alpha$} is defined as
\begin{align*}
\tD_\alpha(\rho\|\sigma) &\coloneqq \frac{1}{\alpha-1}\log \left[(\tr\rho)^{-1} \tr\left\lbrace\left(\sigma^{(1-\alpha)/2\alpha} \rho \sigma^{(1-\alpha)/2\alpha} \right)^\alpha \right\rbrace \right]\\
&= \frac{2\alpha}{\alpha-1} \log \left\| \rho^{1/2}\sigma^{(1-\alpha)/2\alpha}\right\|_{2\alpha} - \frac{1}{\alpha-1}\log\tr\rho.
\end{align*}
Otherwise, we set $\tD_\alpha(\rho\|\sigma)=\infty$.
Note that for $[\rho,\sigma]=0$ the sandwiched Rényi divergence reduces to the usual \emph{$\alpha$-relative Rényi entropy} $D_\alpha(\rho\|\sigma) \coloneqq \frac{1}{\alpha-1}\log \left[(\tr\rho)^{-1}\tr \left(\rho^\alpha \sigma^{1-\alpha}\right)\right]$ (see e.g.~\cite{Pet86}).

\item For $\rho\in\cD(\cH)$ and $\alpha\in(0,\infty)\setminus\lbrace 1\rbrace$, the \emph{Rényi entropy of order $\alpha$} is defined as 
\begin{align*}
S_\alpha(\rho) &\coloneqq -D_\alpha(\rho\|\one) = -\tD_\alpha(\rho\|\one). 
\end{align*}
Note that $S_0(\rho) =\lim_{\alpha\to 0}S_\alpha(\rho)= \log\rk\rho$, where $\rk\rho$ denotes the rank of $\rho$.
We use the notation $S_\alpha(A)_\rho \equiv S_\alpha(\rho_A)$ for $\rho_A\in\cD(\cH_A)$.

\item For $\rho_{AB}\in\cD(\cH_{AB})$ and $\alpha\in(0,\infty)\setminus\lbrace 1\rbrace$, the \emph{Rényi conditional entropy of order $\alpha$} is defined as
\begin{align*}
\tS_\alpha(A|B)_\rho \coloneqq -\min_{\sigma_B} \tD_\alpha(\rho_{AB}\|\one_A\ox\sigma_B), 
\end{align*}
and the \emph{Rényi mutual information of order $\alpha$} is defined as \cite{GW15}
\begin{align*}
\tI_\alpha(A;B)_\rho \coloneqq \min_{\sigma_B} \tD_\alpha(\rho_{AB}\|\rho_A\ox\sigma_B).
\end{align*}
\end{enumerate}
\end{definition}

We have \cite{MDSFT13,WWY14}
\begin{align}\label{eq:D-alpha-limit}
\lim_{\alpha\rightarrow 1}\tD_\alpha(\rho\|\sigma) = D(\rho\|\sigma).
\end{align}
Therefore, in the subsequent discussion of the sandwiched Rényi divergence $\tD_\alpha(\rho\|\sigma)$ and its derived quantities, we will use the full range $\alpha \geq 0$, setting $\tD_1(\rho\|\sigma)=D(\rho\|\sigma)$ and $\tD_0(\rho\|\sigma)=\lim_{\alpha\to 0}\tD_\alpha(\rho\|\sigma)$ \cite{DL14a}. 
In the next proposition we collect a few properties of the Rényi quantities defined above:

\begin{proposition}[\cite{MDSFT13,WWY14}]\label{prop:renyi-properties}
Let $\rho\in\cD(\cH)$ and $\sigma\in\cP(\cH)$. The sandwiched Rényi divergence and entropies derived from it satisfy the following properties:
\begin{enumerate}[{\normalfont (i)}]
\item\label{item:RE-mon-alpha} Monotonicity: If $0 < \alpha \leq\beta$, then $\tD_\alpha(\rho\|\sigma) \leq \tD_\beta(\rho\|\sigma)$.

\item\label{item:RE-dim-bound} Positivity and dimension bound: Let $d=\dim\cH$, then we have $0\leq S_\alpha(\rho)\leq \log d$ for all $\alpha\geq 0$. The extremal values are achieved for pure states and completely mixed states, respectively.

\item\label{item:RE-add} Additivity: For all $\alpha > 0$ and $\rho_i\in\cD(\cH_i)$, $\sigma_i\in\cP(\cH_i)$ for $i=1,2$ such that the conditions on their supports given by \Cref{def:renyi-entropies} hold, we have 
\begin{align*}
\tD_\alpha(\rho_1\ox\rho_2\|\sigma_1\ox\sigma_2) &= \tD_\alpha(\rho_1\|\sigma_1) + \tD_\alpha(\rho_2\|\sigma_2),\\
S_\alpha(\rho_1\ox\rho_2) &= S_\alpha(\rho_1)+S_\alpha(\rho_2).
\end{align*}
Furthermore, additivity also holds for the Rényi conditional entropy and mutual information {\normalfont \cite{Bei13,HT14}}: If $\rho_{AB}$ and $\sigma_{A'B'}$ are bipartite states and $\alpha\geq 1/2$, then  
\begin{align*}
\tS_\alpha(AA'|BB')_{\rho\ox\sigma} &= \tS_\alpha(A|B)_\rho + \tS_\alpha(A'|B')_\sigma,\\
\tI_\alpha(AA';BB')_{\rho\ox\sigma} &= \tI_\alpha(A; B)_\rho + \tI_\alpha(A';B')_\sigma.
\end{align*}

\item\label{item:RE-isom} Invariance under isometries: Let $U\colon \cH\rightarrow\cH'$ be an isometry. 
Then $\tD_\alpha(U\rho U^\dagger\|U\sigma U^\dagger) = \tD_\alpha(\rho\|\sigma)$ and $S_\alpha(U\rho U^\dagger) = S_\alpha(\rho)$ for all $\alpha\geq 0$.

\item\label{item:RE-duality} Duality for the Rényi entropy: Let $|\psi_{AB}\rangle$ be a pure state and set $\rho_A=\tr_B\psi_{AB}$ and $\rho_B=\tr_A\psi_{AB}$. Then $S_\alpha(A)_\rho = S_\alpha(B)_\rho$ for all $\alpha\geq 0$.

\item\label{item:RE-cond-duality} Duality for the Rényi conditional entropy: Let $|\psi_{ABC}\rangle$ be a pure state and for $\alpha\geq 1/2$ define $\beta$ through $1/\alpha + 1/\beta = 2$. Then 
\begin{align*}
\tS_\alpha(A|B)_\psi = -\tS_\beta(A|C)_\psi.
\end{align*}

\item\label{item:dpi} {\normalfont \cite{FL13,Bei13}} Data processing inequality: If $\alpha \in[1/2,\infty)$ and $\Lambda$ is a CPTP map, then
\begin{align*}
\tD_\alpha(\rho\|\sigma) \geq \tD_\alpha(\Lambda(\rho)\|\Lambda(\sigma)).
\end{align*}
Furthermore, let $\Phi\colon B\to C$ be a CPTP map, and for a bipartite state $\rho_{AB}$ set $\sigma_{AC}\coloneqq (\id_A\ox\Phi)(\rho_{AB})$.
We then have for $\alpha\geq 1/2$ that
\begin{align*}
\tS_\alpha(A|B)_\rho &\leq \tS_\alpha(A|C)_\sigma,\\
\tI_\alpha(A;B)_\rho &\geq \tI_\alpha(A;C)_\sigma.
\end{align*}
\end{enumerate}
\end{proposition}

\subsection{Further properties of Rényi entropic quantities}\label{sec:renyi-properties}
In this section we derive various properties of the Rényi entropy, the Rényi conditional entropy, and the Rényi mutual information (see \Cref{def:renyi-entropies}), as well as the Rényi conditional mutual information, defined in \eqref{eq:renyi-cmi} below. 
We defer the proofs of the statements concerning Rényi mutual information and Rényi conditional mutual information to \Cref{sec:renyi-properties-proofs}.

We start with the following subadditivity property for the Rényi entropies \cite{vDH02}, for which we give a simplified proof based on the data processing inequality.

\begin{lemma}[Subadditivity of Rényi entropies {\normalfont \cite{vDH02}}]\label{lem:renyi-subadditivity}
If $\alpha\geq 0$ and $\rho_{AB}\in\cD(\cH_{AB})$, then\footnote{Note that both inequalities in \Cref{lem:renyi-subadditivity} can be tightened by replacing the log terms with the $0$-Rényi entropy $S_0(B)_\rho=\log\rk\rho_B$. 
However, throughout this paper we assume Hilbert spaces to be restricted to the support of the corresponding quantum states, so that $S_0(B)_\rho = \log|B|$.}
\begin{align*}
S_\alpha(A)_\rho - \log|B| \leq S_\alpha(AB)_\rho\leq S_\alpha(A)_\rho + \log|B|.
\end{align*}
\end{lemma}

\begin{proof}
To prove the upper bound on $S_\alpha(AB)_\rho$, observe first that 
\begin{align}
S_{\alpha}(AB)_\rho &= - \tD_\alpha(\rho_{AB}\|\one_A\ox\pi_B) + \log |B|\label{eq:renyi-entropy-as-sand-renyi-div}\\
&= - D_\alpha(\rho_{AB}\|\one_A\ox\pi_B)  + \log |B|.\label{eq:renyi-entropy-as-renyi-div}
\end{align} 
Assuming that $\alpha\geq 1/2$ and using \eqref{eq:renyi-entropy-as-sand-renyi-div}, we then have
\begin{align*}
S_\alpha(AB)_{\rho} &= - \tD_\alpha(\rho_{AB}\|\one_A\ox\pi_B) + \log |B|\\
&\leq - \tD_\alpha(\rho_{A}\|\one_A)  + \log |B|\\
&= S_\alpha(A)_{\rho} + \log |B|,
\end{align*}
where the inequality follows from the data processing inequality (\Cref{prop:renyi-properties}(\ref{item:dpi})). 
If $\alpha \in[0,1/2)$, we use relation \eqref{eq:renyi-entropy-as-renyi-div} together with the data processing inequality for the $\alpha$-Rényi relative entropy \cite{Pet86} instead. 
The lower bound on $S_\alpha(AB)_\rho$ follows from the upper bound by a simple duality argument, using \Cref{prop:renyi-properties}(\ref{item:RE-duality}), as discussed in~\cite{vDH02}.
\end{proof}

We proceed with the following lemma concerning dimension bounds on the Rényi conditional entropy and mutual information, as well as invariance properties with respect to tensor product states.

\begin{lemma}\label{lem:renyi-quantities}
Let $\alpha\in[1/2,\infty)$.
\begin{enumerate}[{\normalfont (i)}]
\item\label{item:RC-dim-bound} For an arbitrary tripartite state $\rho_{ABC}$ we have
\begin{align}
\tS_\alpha(A|BC)_\rho + 2\log|C| &\geq \tS_\alpha(A|B)_\rho,\label{eq:dim-bound-conditional}\\
\tI_\alpha(A;B)_\rho + 2\log|C| &\geq \tI_\alpha(A;BC)_\rho.\label{eq:dim-bound-mutual}
\end{align}

\item\label{item:RC-decoupled} For states $\rho_{AB}$ and $\sigma_C$, we have
\begin{align}
\tS_\alpha(A|BC)_{\rho\ox\sigma} &= \tS_\alpha(A|B)_\rho, \label{eq:product-conditional}\\
\tI_\alpha(A;BC)_{\rho\ox\sigma} &= \tI_\alpha(A;B)_\rho.\label{eq:product-mutual}
\end{align}
\end{enumerate} 
\end{lemma}
\begin{proof}
We first prove \eqref{eq:dim-bound-conditional}. 
By \cite[Prop.~8]{MDSFT13} we have the following bound:
\begin{align}\label{eq:intermediate}
\tS_{\alpha}( A|BC)  _{\rho}\geq \tS_{\alpha
}(  AC|B)  _{\rho}-\log\vert C\vert
\end{align}
Now consider from duality (\Cref{prop:renyi-properties}(\ref{item:RE-cond-duality})) that%
\begin{align*}
\tS_{\alpha}(AC|B )  _{\rho}=-\tS_{\beta
}(  AC|D)  _{\rho},
\end{align*}
where $D$ is a purifying system and $\beta$ is such that $1/\alpha+1/\beta=2$.
By the same reasoning, we find that%
\begin{align*}
\tS_{\beta} ( A|CD)  _{\rho}\geq \tS_{\beta} (  AC|D )_{\rho}-\log\vert C\vert.
\end{align*}
But from duality this is the same as%
\begin{align*}
\tS_{\alpha} (A|B)_{\rho}-\log\vert C\vert \leq\tS_{\alpha}( AC|B)  _{\rho}.
\end{align*}
Substituting this in \eqref{eq:intermediate} then yields the claim.

We continue with the proof of \eqref{eq:product-conditional}. 
From the data processing inequality (\Cref{prop:renyi-properties}(\ref{item:dpi})), we know that%
\begin{align*}
\tS_{\alpha}( A|BC)  _{\rho\ox\sigma}\leq\tS_{\alpha
}(A|B)  _{\rho}.
\end{align*}
On the other hand, consider that%
\begin{align*}
-\tS_{\alpha}( A|BC)  _{\rho\ox\sigma}  &  =\min_{\tau_{BC}%
}\tD_{\alpha}(  \rho_{AB}\otimes\sigma_{C}\Vert \one_{A}%
\otimes\tau_{BC}) \\
&  \leq\tD_{\alpha}(  \rho_{AB}\otimes\sigma_{C}\Vert
\one_{A}\otimes\theta_B\otimes\sigma_{C}) \\
&  =\tD_{\alpha}( \rho_{AB}\Vert \one_{A}\theta
_{B})  .
\end{align*}
Since the inequality holds for all $\theta_{B}$, we get that%
\begin{align*}
-\tS_{\alpha}( A|BC)  _{\rho\ox\sigma}\leq- \tS_{\alpha}\left(  A|B\right)  _{\rho},
\end{align*}
which yields \eqref{eq:product-conditional}. 

The corresponding relations \eqref{eq:dim-bound-mutual} and \eqref{eq:product-mutual} for the Rényi mutual information are proved in \Cref{sec:renyi-properties-proofs}.
\end{proof}

The following proposition is crucial for our proofs. It bounds the difference of Rényi entropic quantities of two quantum states in terms of their fidelity. 
Note that the inequality \eqref{eq:renyi-fidelity-entropy} for the Rényi entropies originally appeared in \cite{vDH02}. 
The last assertion concerns the Rényi conditional mutual information $\tI_\alpha(A;B|C)_\rho$, defined in \cite{BSW15} for a tripartite state $\rho_{ABC}$ and $\alpha>0$ as
\begin{align}\label{eq:renyi-cmi}
\tI_\alpha(A;B|C)_\rho = \frac{2\alpha}{\alpha-1} \log \left\| \rho_{ABC}^{1/2}\,\rho_{AC}^{(1-\alpha)/2\alpha}\rho_C^{(\alpha-1)/2\alpha} \rho_{BC}^{(1-\alpha)/2\alpha} \right\|_{2\alpha}.
\end{align}
Note that this quantity does not feature in our proofs.
However, the corresponding fidelity bound in \eqref{eq:renyi-fidelity-cmi} might be of independent interest, and we include it for the sake of completeness.
\begin{proposition}\label{prop:renyi-fidelity}
Let $\rho_{AB},\sigma_{AB}\in\cD(\cH_{AB})$, and for $\alpha\in(1/2,1)$ define $\beta\equiv \beta(\alpha) \coloneqq \alpha/(2\alpha-1)$. Then the following inequalities hold:
\begin{align}
S_\alpha(A)_\rho - S_{\beta}(A)_\sigma &\geq \frac{2\alpha}{1-\alpha} \log F(\rho_A,\sigma_A),\label{eq:renyi-fidelity-entropy}\\
\tS_\alpha(A|B)_\rho - \tS_\beta(A|B)_\sigma &\geq \frac{2\alpha}{1-\alpha} \log F(\rho_{AB},\sigma_{AB}).\label{eq:renyi-fidelity-conditional}\\
\intertext{Assuming that $\rho_A=\sigma_A$ holds, we also have}
\tI_\beta(A;B)_\rho - \tI_\alpha (A;B)_\sigma &\geq \frac{2\alpha}{1-\alpha} \log F(\rho_{AB},\sigma_{AB}).\label{eq:renyi-fidelity-mutual}
\intertext{Let $\rho_{ABC}$ and $\sigma_{ABC}$ be tripartite states satisfying $\rho_{AC}=\sigma_{AC}$, $\rho_{BC}=\sigma_{BC}$, and $\rho_C=\sigma_C$. Furthermore, assume that $\rho_{BC}$ has full support. Then}
\tI_\beta(A;B|C)_\rho - \tI_\alpha(A;B|C)_\sigma &\geq \frac{2\alpha}{1-\alpha} \log F(\rho_{ABC},\sigma_{ABC}).\label{eq:renyi-fidelity-cmi}
\end{align}
\end{proposition}

\begin{proof}
We first observe that \eqref{eq:renyi-fidelity-entropy} follows from \eqref{eq:renyi-fidelity-conditional} by setting $B=\mathbb{C}$. 
To prove \eqref{eq:renyi-fidelity-conditional}, let $\tau_{B}$ be an arbitrary density operator, and let $\varepsilon\in\left(  0,1\right) $. 
Furthermore, let
\begin{align*}
\tau(  \eps) _{B}\coloneqq (1-\eps)  \tau
_{B}+\eps\pi_{B},
\end{align*}
and recall that
\begin{align}
\tD_{\alpha}(  \omega\Vert\theta)  \leq\tD_{\alpha}(  \omega\Vert\theta^{\prime})\label{eq:domination-relation}%
\end{align}
holds for all $\alpha\in\left[  1/2,\infty\right] $ and for $\theta\geq \theta^{\prime}\geq0$ \cite[Prop.~4]{MDSFT13}. 
Observe also that for $c>0$ we have 
\begin{align}
\tD_{\alpha}(  \omega\Vert c\theta)  =\tD_{\alpha}(  \omega\Vert\theta)  -\log c.\label{eq:c-relation}%
\end{align}
Consider then the following chain of inequalities:
\begin{align*}
& -\tD_{\alpha}(  \rho_{AB}\Vert I_{A}\otimes
\tau(\eps)  _{B}) +\tD_{\beta
}(  \sigma_{AB}\Vert I_{A}\otimes\tau_{B})  -\log\left(
1-\eps\right)  \nonumber\\
&\qqquad =-\tD_{\alpha}(  \rho_{AB}\Vert I_{A}\otimes\tau(
\eps)  _{B})  +\tD_{\beta}(  \sigma
_{AB}\Vert I_{A}\otimes\left(  1-\eps\right)  \tau_{B})  \\
&\qqquad \geq-\tD_{\alpha}(  \rho_{AB}\Vert I_{A}\otimes\tau(
\eps)  _{B})  +\tD_{\beta}(  \sigma
_{AB}\Vert I_{A}\otimes\tau(  \eps)  _{B})  \\
&\qqquad =\frac{2\alpha}{1-\alpha}\log\left\Vert \rho_{AB}^{1/2}\tau(\eps)  _{B}^{\left(  1-\alpha\right)  /2\alpha}\right\Vert
_{2\alpha}+\frac{2\beta}{\beta-1}\log\left\Vert \tau(\eps)  _{B}^{\left(  1-\beta\right)  /2\beta}\sigma_{AB}^{1/2}\right\Vert
_{2\beta}\\
&\qqquad =\frac{2\alpha}{1-\alpha}\log\left[  \left\Vert \rho_{AB}^{1/2}\tau(
\eps)  _{B}^{\left(  1-\alpha\right)  /2\alpha}\right\Vert
_{2\alpha}\left\Vert \tau(  \eps)  _{B}^{\left(
1-\beta\right)  /2\beta}\sigma_{AB}^{1/2}\right\Vert _{2\beta}\right]  \\
&\qqquad \geq\frac{2\alpha}{1-\alpha}\log  \left\Vert \rho_{AB}^{1/2}%
\tau(\eps)  _{B}^{\left(  1-\alpha\right)  /2\alpha}%
\tau(  \eps)  _{B}^{\left(  1-\beta\right)  /2\beta}%
\sigma_{AB}^{1/2}\right\Vert _{1}  \\
&\qqquad =\frac{2\alpha}{1-\alpha}\log  \left\Vert \rho_{AB}^{1/2}\sigma
_{AB}^{1/2}\right\Vert _{1}  \\
&\qqquad =\frac{2\alpha}{1-\alpha}\log F(  \rho_{AB},\sigma_{AB})  .\numberthis\label{eq:intermediate-result}
\end{align*}
The first equality is an application of (\ref{eq:c-relation}). 
The first inequality is a consequence of (\ref{eq:domination-relation}) and the fact that%
\begin{align*}
\left(  1-\eps\right)  \tau_{B}\leq\left(  1-\eps\right)
\tau_{B}+\eps\pi_{B}=\tau(  \eps)  _{B}.
\end{align*}
The second and third equality follow from the definition of the sandwiched Rényi divergence (see \Cref{def:renyi-entropies}) and the relation 
\begin{align}\label{eq:beta-alpha-relation}
\frac{\beta}{\beta-1}  =\frac{\alpha}{1-\alpha}.
\end{align} 
The second inequality is an application of H\"older's inequality \eqref{eq:hoelder-inequality} (note that $1/2\alpha +1/2\beta=1$). 
The second-to-last equality follows because $\tau( \eps)  _{B}$ is a full rank operator for $\eps\in\left(0,1\right)  $, so that $\tau(\eps)  _{B}^{\left(1-\alpha\right) /2\alpha}\tau(  \eps)  _{B}^{\left(1-\beta\right)  /2\beta}=\one_{B}$. 

Since \eqref{eq:intermediate-result} holds for an arbitrary density operator $\tau_B$, we may choose $\tau_B$ to be the optimizing state for $\tS_\beta(A|B)_\sigma$.
We can then continue from \eqref{eq:intermediate-result} as
\begin{align*}
\frac{2\alpha}{1-\alpha}\log F(  \rho_{AB},\sigma_{AB})    &
\leq-\tD_{\alpha}(  \rho_{AB}\Vert I_{A}\otimes\tau(\eps)  _{B})  +\tD_{\beta}(  \sigma
_{AB}\Vert I_{A}\otimes\tau_{B})  -\log\left(  1-\eps\right)  \\
& \leq  \max_{\omega_{B}\in\mathcal{D}\left(  \mathcal{H}_{B}\right)
}\left\{ -\tD_{\alpha}(  \rho_{AB}\Vert I_{A}\otimes\omega_{B}) \right\}
  +\tD_{\beta}(  \sigma_{AB}\Vert I_{A}\otimes\tau
_{B})  -\log\left(  1-\eps\right)  \\
& =\tS_{\alpha}(  A|B)  _{\rho} - \tS_\beta(A|B)_\sigma  -\log\left(
1-\eps\right).
\end{align*}
We have shown that this relation holds for all $\eps\in\left(0,1\right)  $, and so taking the limit $\eps\searrow 0$ yields the claim.

The bounds \eqref{eq:renyi-fidelity-mutual} and \eqref{eq:renyi-fidelity-cmi} are proved in \Cref{sec:renyi-properties-proofs}.
\end{proof}

The next result concerns the Rényi conditional entropy of c-q states. 
Note that a special case of \eqref{eq:discarding-classical-info} for the Rényi entropies (i.e.~where system $B$ is trivial) appeared in \cite{Sha14}. 

\begin{proposition}\label{prop:renyi-cq-states}
Let $\rho_{ABX}=\sum_{x\in\cX} p_x \rho_{AB}^x\ox |x\rangle\langle x|_X$ be a c-q state. 
Then the following properties hold for all $\alpha>0$:
\begin{enumerate}[{\normalfont (i)}]
\item Monotonicity under discarding classical information:
\begin{align}\label{eq:discarding-classical-info}
\tS_\alpha(AX|B)_\rho \geq \tS_\alpha(A|B)_\rho.
\end{align}

\item Dimension bound:
\begin{align}
\tS_\alpha(A|BX)_\rho + \log |X| &\geq \tS_\alpha(A|B)_\rho,\label{eq:classical-dim-bound}\\
\tI_\alpha(A;BX)_\rho &\leq \log|X| + \tI_\alpha(A;B)_\rho.\label{eq:classical-dim-bound-mutual}
\end{align}
\end{enumerate}
\end{proposition}

\begin{proof}
To prove \eqref{eq:discarding-classical-info}, let $\tau_B$ be the optimizing state for $\tS_\alpha(A|B)_\rho$, and assume $\alpha\neq 1$. We then have
\begin{align*}
\tS_\alpha(AX|B)_\rho & \geq \frac{1}{1-\alpha} \log\tr\left[\left(\tau_B^{(1-\alpha)/2\alpha}\rho_{ABX}\tau_B^{(1-\alpha)/2\alpha}\right)^\alpha\right]\\
&= \frac{1}{1-\alpha} \log\left\lbrace\sumi_{x\in\cX} \tr\left[\left(\tau_B^{(1-\alpha)/2\alpha}p_x\rho_{AB}^x\tau_B^{(1-\alpha)/2\alpha}\right)^\alpha\right]\right\rbrace\\
&= \frac{1}{1-\alpha} \log\left\lbrace\sumi_{x\in\cX} \left\|\tau_B^{(1-\alpha)/2\alpha}p_x\rho_{AB}^x\tau_B^{(1-\alpha)/2\alpha}\right\|_\alpha^\alpha\right\rbrace\\
&\geq \frac{1}{1-\alpha} \log\left\lbrace \left\|\tau_B^{(1-\alpha)/2\alpha}\sumi_{x\in\cX} p_x\rho_{AB}^x\tau_B^{(1-\alpha)/2\alpha}\right\|_\alpha^\alpha\right\rbrace\\
&= \frac{1}{1-\alpha} \log\left\lbrace \left\|\tau_B^{(1-\alpha)/2\alpha}\rho_{AB}\tau_B^{(1-\alpha)/2\alpha}\right\|_\alpha^\alpha\right\rbrace\\
&= \tS_\alpha(A|B)_\rho,
\end{align*}
where the inequality follows from McCarthy's inequalities, using \eqref{eq:mccarthy-quasinorm} for $\alpha <1$ and \eqref{eq:mccarthy-norm} for $\alpha>1$. 
For $\alpha = 1$, the claim follows easily from definition \eqref{eq:quantum-relative-entropy} of the quantum relative entropy.

To prove \eqref{eq:classical-dim-bound}, observe that we have 
\begin{align*}
\tS_\alpha(A|BX)_\rho + \log|X| \geq \tS_\alpha(AX|B)_\rho \geq \tS_\alpha(A|B)_\rho,
\end{align*}
where the first inequality is \cite[Prop.~8]{MDSFT13}, and the second inequality is \eqref{eq:discarding-classical-info}.

Finally, \eqref{eq:classical-dim-bound-mutual} follows from \eqref{eq:classical-dim-bound} and the reasoning in \Cref{sec:renyi-properties-proofs-sub}.
\end{proof}

\section{State redistribution}\label{sec:state-re}

\subsection{The protocol}\label{sec:state-re-protocol}
Consider a tripartite state $\rho_{ABC}$ shared between Alice and Bob, with the systems $A$ and $C$ being with Alice and the system $B$ being with Bob.
Let $\psi_{ABCR}$ denote a purification of $\rho_{ABC}$, with $R$ being an inaccessible, purifying reference system. 
Furthermore, Alice and Bob share entanglement in the form of an MES $\Phi_{T_AT_B}^k$ of Schmidt rank $k$, with the systems $T_A$ and $T_B$ being with Alice and Bob, respectively. 
The goal of the state redistribution protocol is to transfer the system $A$ from Alice to Bob, while preserving its correlations with the other systems.
In the process, the shared entanglement is transformed to an MES $\Phi_{T_A'T_B'}^m$ of Schmidt rank $m$, where $T_A'$ and $T_B'$ are with Alice and Bob, respectively. 
The initial state and the target state are shown in \Cref{fig:state-re-schematic}.

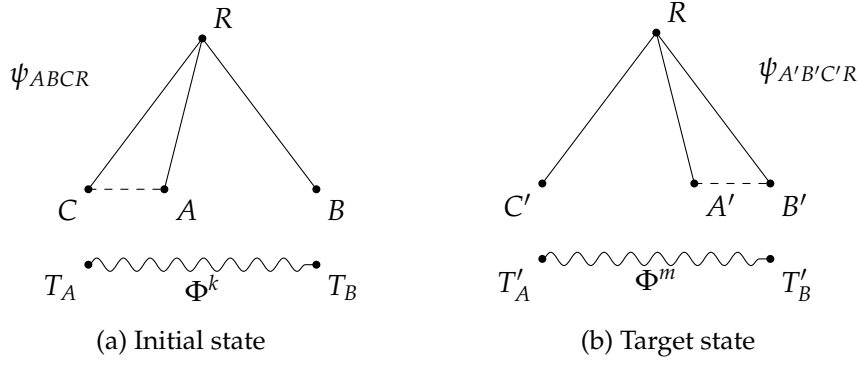
\begin{figure}[t]
\centering
\begin{subfigure}[b]{0.3\textwidth}
\begin{tikzpicture}
\newcommand{\yu}{-2} 
\newcommand{\yuu}{-3} 
\node at (-0.5,-0.5) {$\psi_{ABCR}$};
\coordinate (R1) at (1.5,0);
\fill (R1) circle[radius=1.5pt] node[above right] {$R$};
\coordinate (C1) at (0,\yu);
\fill (C1) circle[radius=1.5pt] node[below left] {$C$};
\coordinate (A) at (1,\yu);
\fill (A) circle[radius=1.5pt] node[below right] {$A$};
\coordinate (B1) at (3,\yu);
\fill (B1) circle[radius=1.5pt] node[below right] {$B$};
\coordinate (A0) at (0,\yuu);
\fill (A0) circle[radius=1.5pt] node[below left] {$T_A$};
\coordinate (B0) at (3,\yuu);
\fill (B0) circle[radius=1.5pt] node[below right] {$T_B$};

\draw (R1) -- (C1);
\draw (R1) -- (A);
\draw (R1) -- (B1);
\draw[dashed] (C1) -- (A);
\draw[snake it] (A0) -- (B0) node[below,pos=0.5] {$\Phi^k$};
\end{tikzpicture}
\caption{Initial state}
\label{fig:init-state}
\end{subfigure}
\qquad\qquad
\begin{subfigure}[b]{0.3\textwidth}
\begin{tikzpicture}
\newcommand{\yu}{-2} 
\newcommand{\yuu}{-3} 
\node at (3.5,-0.5) {$\psi_{A'B'C'R}$};
\coordinate (R1) at (1.5,0);
\fill (R1) circle[radius=1.5pt] node[above right] {$R$};
\coordinate (C1) at (0,\yu);
\fill (C1) circle[radius=1.5pt] node[below left] {$C'$};
\coordinate (Bp) at (2,\yu);
\fill (Bp) circle[radius=1.5pt] node[below right] {$A'$};
\coordinate (B1) at (3,\yu);
\fill (B1) circle[radius=1.5pt] node[below right] {$B'$};
\coordinate (A0) at (0,\yuu);
\fill (A0) circle[radius=1.5pt] node[below left] {$T_A'$};
\coordinate (B0) at (3,\yuu);
\fill (B0) circle[radius=1.5pt] node[below right] {$T_B'$};

\draw (R1) -- (C1);
\draw (R1) -- (Bp);
\draw (R1) -- (B1);
\draw[dashed] (Bp) -- (B1);
\draw[snake it] (A0) -- (B0) node[below,pos=0.5] {$\Phi^m$};
\end{tikzpicture}
\caption{Target state}
\label{fig:target-state}
\end{subfigure}
\caption{State redistribution protocol that transfers Alice's system $A$ to Bob. 
Starting with the initial state $\psi\ox\Phi^k$ depicted in (\subref{fig:init-state}), the protocol outputs a state that is close in fidelity to the target state $\psi\ox\Phi^m$ depicted in (\subref{fig:target-state}).}
\label{fig:state-re-schematic}
\end{figure}

In achieving this goal, Alice and Bob are allowed to use local encoding and decoding operations on the systems in their possession. 
In addition, Alice is allowed to send qubits to Bob (through a noiseless quantum channel). 
A general state redistribution protocol $(\rho,\Lambda)$ with $\Lambda \equiv \cD \circ \cE$ and $\rho=\rho_{ABC}$ consists of the following steps (cf.~\Cref{fig:state-re-protocol}):
\begin{enumerate} 
\item Alice applies an encoding CPTP map $\cE\colon ACT_A\rightarrow C'T_A'Q$ and sends the system $Q$ to Bob through the noiseless quantum channel.

\item Upon receiving the system $Q$, Bob applies a decoding CPTP map $\cD\colon QBT_B\rightarrow T_B'A'B',$ where $T_B'\cong T_A'$ and $A'\cong A$. 
\end{enumerate}
The initial state shared between Alice, Bob, and the reference is $\Phi_{T_AT_B}^k\ox\psi_{ABCR}$, the state after Alice's encoding operation is
\begin{align}\label{eq:state-re-encoded-state}
\omega\equiv \omega_{T_A'T_BC'QBR}\coloneqq (\cE\ox\id_R)(\Phi^k\ox\psi),
\end{align}
and the final state of the protocol $(\rho,\Lambda)$ is given by
\begin{align}\label{eq:state-re-final-state}
\sigma \equiv \sigma_{T_A'T_B'A'B'C'R}\coloneqq (\Lambda\ox\id_R)(\Phi^k\ox\psi) = (\cD\circ\cE\ox\id_R)(\Phi^k\ox\psi).
\end{align}
The aim is to obtain a state $\sigma$ that is close to the target state $\Phi_{T_A'T_B'}^m\ox \psi_{A'B'C'R}$, where $\psi_{A'B'C'R}=\psi_{ABCR}$. 
The figure of merit of the protocol is the fidelity $F(\sigma, \Phi^m\ox\psi).$
The number of qubits that Alice sends to Bob, is given by $\log |Q|$, whereas the number of ebits consumed in the protocol is given by $\log k-\log m= \log |T_A| - \log |T_A'|$.
If $k<m$ then ebits are \emph{gained} in the protocol.

\begin{figure}[t]
\centering
\includegraphics[width=0.7\textwidth]{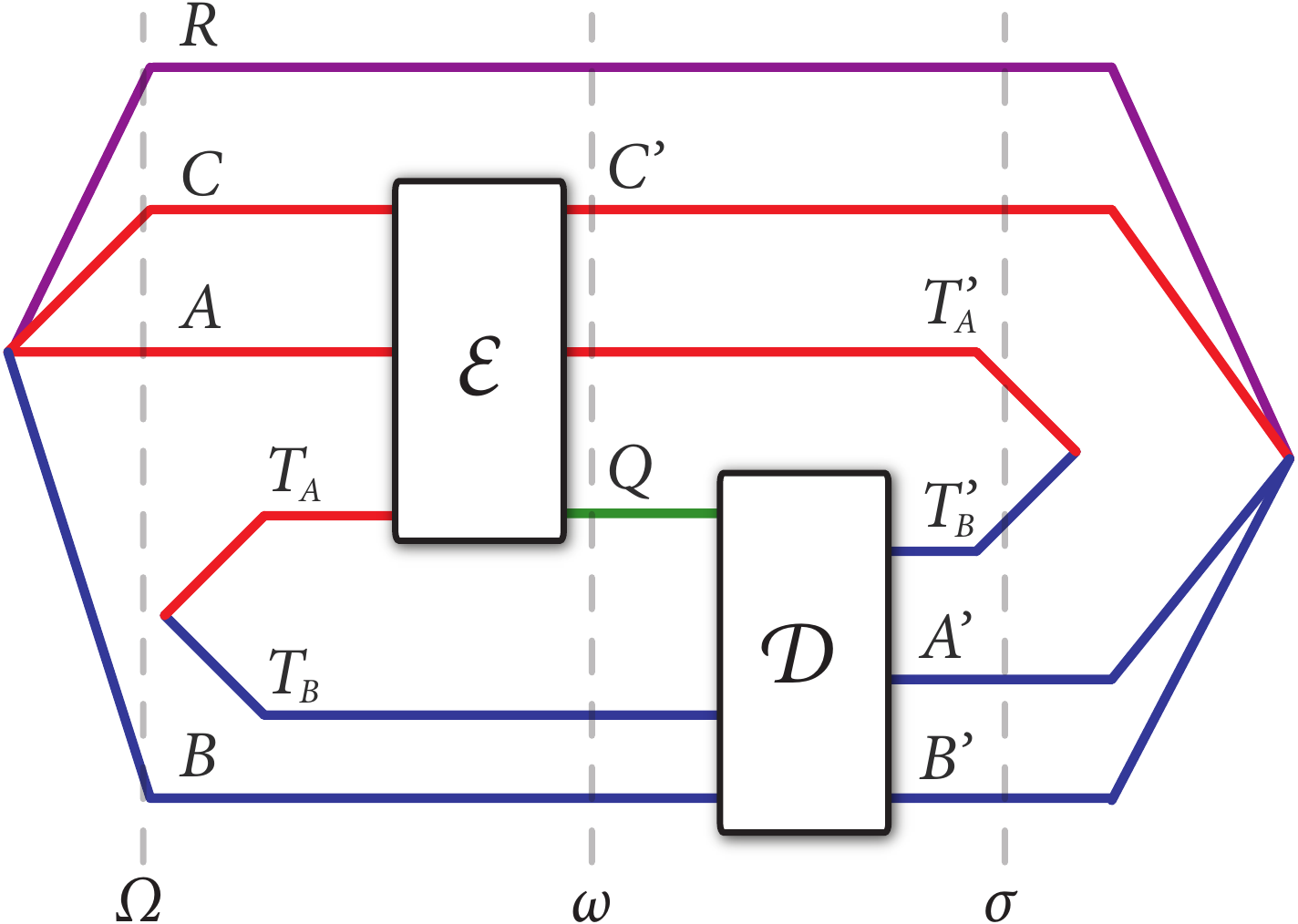}
\caption{State redistribution protocol (see \Cref{sec:state-re-protocol} for a detailed explanation).}
\label{fig:state-re-protocol}
\end{figure}

We consider state redistribution (and all other tasks studied in this paper) in the asymptotic, memoryless setting, where Alice and Bob start with multiple (say, $n$) copies of the initial resource, and the strong converse property is established in the limit $n\rightarrow \infty$.
In this case, Alice and Bob initially share $n$ identical copies of the state $\rho_{ABC}$ with purification $\psi_{ABCR}$, i.e.~they share the state $\rho_{ABC}^{\otimes n}$ with purification $\psi^{\ox n}_{ABCR}$.
Moreover, they share an MES $\Phi^{k_n}_{T_A^nT_B^n}$ of Schmidt rank $k_n$.
We then consider state redistribution protocols $(\rho^{\ox n},\Lambda_n)$ with the figure of merit
\begin{align}\label{eq:state-re-figure-of-merit}
F_n \coloneqq F\big(\sigma_n, \Phi^{m_n}\ox\psi^{\ox n} \big),
\end{align}
where $\Phi^{m_n}_{T_A'^nT_B'^n}$ is an MES of Schmidt rank $m_n$, and $\sigma_n \coloneqq (\Lambda_n\ox\id_{R^n}) (\Phi^{k_n}\ox\psi^{\ox n})$, where $\Lambda_n\colon A^nC^nT_A^n\ox B^nT_B^n\rightarrow C'^nT_A'^n\ox T_B'^nA'^nB'^n$ with $Q^n$ being sent from Alice to Bob.
The two operational quantities of interest are as follows:
\begin{enumerate}
\item The quantum communication cost of the protocol $(\rho^{\ox n},\Lambda_n)$, given by
\begin{align}\label{eq:state-re-qu-comm-cost}
q(\rho^{\ox n}, \Lambda_n) \coloneqq \frac{1}{n} \log |Q^n|.
\end{align}
\item The entanglement cost of the protocol $(\rho^{\ox n},\Lambda_n)$, given by 
\begin{align}\label{eq:state-re-ent-cost}
e(\rho^{\ox n}, \Lambda_n) \coloneqq \frac{1}{n} \left(\log k_n - \log m_n\right) = \frac{1}{n} \left(\log |T_A^n| - \log |T_A'^n|\right).
\end{align}
\end{enumerate}
A pair $(e,q)\in\mathbb{R}^2$ with $q \geq 0$, is said to be an achievable rate pair for state redistribution of a state $\rho_{ABC}$, if there exists a sequence of protocols $\{(\rho^{\ox n},\Lambda_n)\}_{n \in \mathbb{N}}$, satisfying 
\begin{align*}
\limsup_{n \to \infty} q(\rho^{\ox n}, \Lambda_n) &= q, & \limsup_{n \to \infty} e(\rho^{\ox n},\Lambda_n) &= e, & \liminf_{n \to \infty} F_n &= 1.
\end{align*}

Luo and Devetak \cite{LD09} and Yard and Devetak \cite{YD09} (see also Devetak and Yard \cite{DY08}) proved that a pair $(e,q)$ is an achievable rate pair for state redistribution of a state $\rho_{ABC}$ 
if and only if it lies in the region (cf.~\Cref{fig:achievable-rate-region}) defined by
\begin{align}\label{eq:state-re-optimal-rates}
q &\geq \frac{1}{2} I(A;R|B)_\rho, & \quad q + e &\geq S(A|B)_\rho.
\end{align}
Recently, a strong converse theorem for the quantum communication cost was proved by Berta {\em{et al.}} \cite{BCT14v2}, using the smooth entropy framework (cf.~\cite{Ren05,TCR09,Dat09,TCR10} and references therein). 
This theorem, however, did not prove the strong converse property for the entire boundary of the achievable rate region given by \eqref{eq:state-re-optimal-rates}. 
We fill this gap with \Cref{thm:state-re-strong-converse}, as well as provide an alternative proof of the strong converse theorem of \cite{BCT14v2}. 
As mentioned earlier, Berta \etal have now also extended their proof to the entire  boundary of the achievable rate region \cite{BCT14}.

\subsection{Strong converse theorem}\label{sec:state-re-strong-converse}

\begin{lemma}\label{lem:state-re-converse}
Let $\rho\equiv\rho_{ABC}$ be a tripartite state with purification $|\psi_{ABCR}\rangle$, and let $(\rho,\Lambda)$ be a state redistribution protocol where $\Lambda\equiv\cD\circ\cE$ with $\cE\colon ACT_A\rightarrow C'T_A'Q$ and $\cD\colon QBT_B\rightarrow T_B'A'B'$ as defined in \Cref{sec:state-re-protocol}. 
Furthermore, set
\begin{align*}
F &\coloneqq F\left(\Phi_{T_A'T_B'}^m\ox \psi_{A'B'C'R}, (\cD\circ\cE\ox\id_R)\left(\Phi_{T_AT_B}^k\ox\psi_{ABCR}\right)\right).
\end{align*}
Then we have the following bounds on $F$ for $\alpha\in(1/2,1)$ and $\beta \equiv \beta(\alpha)= \alpha/(2\alpha-1)$:
\begin{align}
\log F &\leq \frac{1-\alpha}{2\alpha} \left( \log|Q| + \log|T_A| - \log|T_A'| - S_\beta(AB)_\rho + S_\alpha(B)_\rho \right),\label{eq:state-re-q+e}\\
\log F &\leq \frac{1-\alpha}{2\alpha} \left( 2\log|Q| - \tS_\beta(R|B)_\rho + \tS_\alpha(R|AB)_\rho \right).\label{eq:state-re-q}\\
\intertext{We also have the following alternative bound to \eqref{eq:state-re-q}:}
\log F &\leq \frac{1-\alpha}{2\alpha} \left( 2\log|Q| - \tI_\alpha(R;AB)_\rho + \tI_\beta(R;B)_\rho \right).\label{eq:state-re-q-alt}
\end{align}
\end{lemma}

\begin{proof}
We first prove \eqref{eq:state-re-q+e}. 
Denote by $U_\cE\colon \cH_{ACT_A}\rightarrow\cH_{C'T_A'QE_1}$ and $U_\cD\colon \cH_{QBT_B}\rightarrow \cH_{T_B'A'B'E_2}$ the Stinespring isometries of the maps $\cE$ and $\cD$ respectively, and define the pure states
\begin{align}
|\omega_{C'T_A'QT_BBRE_1} \rangle &\coloneqq U_\cE \left(\big|\Phi_{T_AT_B}^k\big\rangle\ox |\psi_{ABCR}\rangle\right),\label{eq:omega}\\
|\sigma_{T_A'T_B'A'B'C'RE_1E_2}\rangle &\coloneqq U_\cD |\omega_{C'T_A'QT_BBRE_1}\rangle,\label{eq:sigma}
\end{align}
that purify the mixed states $\omega$ and $\sigma$ defined in \eqref{eq:state-re-encoded-state} and \eqref{eq:state-re-final-state}, respectively.
We then have 
\begin{align}
S_\alpha(QT_BB)_\omega \leq \log|Q| + \log |T_A| + S_\alpha(B)_\rho,\label{eq:state-re-dim-bound}
\end{align}
where we used the subadditivity of the Rényi entropies (\Cref{lem:renyi-subadditivity}) twice, as well as the fact that $T_A$ is the same size as $T_B$. 
For the fidelity $F \coloneqq F(\Phi_{T_A'T_B'}^m\ox \psi_{A'B'C'R}, \sigma_{T_A'T_B'A'B'C'R})$, we know by Uhlmann's theorem that there exists a pure state $\phi_{E_1E_2}$ such that
\begin{align*}
F &= F\left(\Phi_{T_A'T_B'}^m\ox \psi_{A'B'C'R} \ox \phi_{E_1E_2}, \sigma_{T_A'T_B'A'B'C'RE_1E_2}\right)\\
&\leq F(\pi^m_{T_B'}\ox\rho_{A'B'}\ox\phi_{E_2},\sigma_{T_B'A'B'E_2}),
\end{align*}
where $|\sigma_{T_A'T_B'A'B'C'RE_1E_2}\rangle$ is the pure state defined in \eqref{eq:sigma}. 
The inequality follows from the monotonicity of the fidelity under partial trace. 
Hence, by \cref{eq:renyi-fidelity-entropy} of \Cref{prop:renyi-fidelity} we obtain the following bound, setting $\beta=\alpha/(2\alpha-1)$:
\begin{align}
S_\alpha(QT_BB)_\omega &= S_\alpha(T_B'A'B'E_2)_\sigma\notag\\
&\geq S_\beta(T_B'A'B'E_2)_{\pi^m\ox\rho\ox\phi} + \frac{2\alpha}{1-\alpha}\log F\notag\\
&\geq \log|T_A'| + S_\beta(A'B')_\rho + \frac{2\alpha}{1-\alpha}\log F,\label{eq:state-re-fidelity-bound}
\end{align}
where we used the invariance of the Rényi entropies under the isometry $U_\cD$ (\Cref{prop:renyi-properties}) in the first equality, \cref{eq:renyi-fidelity-entropy} of \Cref{prop:renyi-fidelity} in the first inequality, and additivity and positivity of the Rényi entropies (\Cref{prop:renyi-properties}(\ref{item:RE-add}) and (\ref{item:RE-dim-bound})) in the second inequality. 
Combining \eqref{eq:state-re-dim-bound} and \eqref{eq:state-re-fidelity-bound} then yields
\begin{align*} 
\log|Q| + \log|T_A| + S_\alpha(B)_\rho &\geq \log|T_A'| + S_\beta(AB)_\rho + \frac{2\alpha}{1-\alpha}\log F,
\end{align*}
which is equivalent to \eqref{eq:state-re-q+e}.

To prove \eqref{eq:state-re-q}, we first note that from \cref{eq:renyi-fidelity-conditional} of \Cref{prop:renyi-fidelity} we have the inequality
\begin{align}\label{eq:cond-bound}
\frac{2\alpha}{1-\alpha} \log F(\rho_{ABR},\sigma_{A'B'R}) \leq \tS_\alpha(R|AB)_\rho - \tS_\beta(R|A'B')_\sigma.
\end{align}
We bound the second term on the right-hand side of \eqref{eq:cond-bound} as follows:
\begin{align*}
- \tS_\beta(R|A'B')_\sigma &\leq - \tS_\beta(R|QBT_B)_\omega\\
&\leq - \tS_\beta(R|BT_B)_\omega + 2\log \vert Q\vert\\
&= - \tS_\beta(R|B)_\omega + 2\log \vert Q\vert\\
&= - \tS_\beta(R|B)_\rho + 2\log \vert Q\vert,\numberthis\label{eq:2nd-term-bound}
\end{align*}
where we used data processing (\Cref{prop:renyi-properties}(\ref{item:dpi})) in the first inequality, and \cref{eq:dim-bound-conditional} and \eqref{eq:product-conditional} of \Cref{lem:renyi-quantities} in the second inequality and the first equality, respectively. 
Substituting \eqref{eq:2nd-term-bound} in \eqref{eq:cond-bound} now yields \eqref{eq:state-re-q}.

The bound \eqref{eq:state-re-q-alt} follows from similar arguments as those used for the proof of \eqref{eq:state-re-q}, relying on \cref{eq:dim-bound-mutual} and \eqref{eq:product-mutual} of \Cref{lem:renyi-quantities} and \cref{eq:renyi-fidelity-mutual} of \Cref{prop:renyi-fidelity} instead. 
We therefore omit an explicit proof.
\end{proof}

\Cref{lem:state-re-converse} immediately implies the following strong converse theorem:

\begin{theorem}[Strong converse for state redistribution]\label{thm:state-re-strong-converse}
Let $\rho\equiv\rho_{ABC}$ be a tripartite state and $\lbrace(\rho^{\ox n},\Lambda_n)\rbrace_{n\in\mathbb{N}}$ be a sequence of state redistribution protocols as described in \Cref{sec:state-re-protocol}, with figure of merit $F_n$ as defined in \eqref{eq:state-re-figure-of-merit}.
Then for all $n\in\mathbb{N}$ we have the following bounds on $F_n$ for $\alpha\in(1/2,1)$ and $\beta=\alpha/(2\alpha-1)$:
\begin{align}
F_n &\leq \exp\left\lbrace - n\kappa(\alpha) \left[S_\beta(AB)_\rho - S_\alpha(B)_\rho - \left(q+e\right) \right]\right\rbrace,\label{eq:state-re-q+e-n}\\
F_n &\leq \exp\left\lbrace - n\kappa(\alpha) \left[\tS_\beta(R|B)_\rho - \tS_\alpha(R|AB)_\rho - 2q\right] \right\rbrace,\label{eq:state-re-q-n}
\end{align}
where $\kappa(\alpha)=(1-\alpha)/(2\alpha)$, and $q\equiv q(\rho^{\ox n},\Lambda_n)$ and $e\equiv e(\rho^{\ox n},\Lambda_n)$ are the quantum communication cost and entanglement cost defined in \eqref{eq:state-re-qu-comm-cost} and \eqref{eq:state-re-ent-cost}, respectively. 
As an alternative to \eqref{eq:state-re-q-n}, we also obtain the bound
\begin{align}
F_n &\leq \exp\left\lbrace - n\kappa(\alpha) \left[\tI_\alpha(R;AB)_\rho - \tI_\beta(R;B)_\rho - 2q\right] \right\rbrace. \label{eq:state-re-q-n-alt}
\end{align}
\end{theorem}

We have 
\begin{subequations}\label{eq:state-re-renyi-convergence}
\begin{align}
S_\beta(AB)_\rho - S_\alpha(B)_\rho &\xrightarrow{\alpha\rightarrow 1} S(A|B)_\rho,\\ 
\tS_\beta(R|B)_\rho - \tS_\alpha(R|AB)_\rho &\xrightarrow{\alpha\rightarrow 1} I(A;R|B)_\rho,\label{eq:state-re-renyi-convergence-b}
\end{align}
\end{subequations}
where \eqref{eq:state-re-renyi-convergence-b} follows from a similar argument as in \cite[Lem.~10]{CMW14}. 
Moreover, the right-hand sides of \eqref{eq:state-re-renyi-convergence} determine the boundary of the region of achievable rate pairs $(e,q)$ given by \eqref{eq:state-re-optimal-rates}. 
Hence, we obtain the following strong converse theorem along the same lines as at the end of \Cref{sec:strong-converse-property}: 
If $q+e< S(A|B)_\rho$ or $q < \onehalf I(A;R|B)_\rho$, there is a constant $K>0$ such that
\begin{align*}
F_n \leq \exp(-nK).
\end{align*}
For the remainder of the paper, we will skip this last step, and merely state strong converse theorems in the form of \Cref{thm:state-re-strong-converse}.

\subsection{Rényi generalizations of the conditional mutual information}

We can regard the expressions $\tS_\beta(R|B)_\rho - \tS_\alpha(R|AB)_\rho$ and $\tI_\alpha(R;AB)_\rho - \tI_{\beta}(R;B)_\rho$ appearing in the bounds on the fidelity in \Cref{lem:state-re-converse} and \Cref{thm:state-re-strong-converse} as Rényi generalizations of the conditional mutual information $I(A;R|B)_\rho$ (see \cite{BSW15} for a detailed discussion of this concept). 
More generally, for a tripartite state $\rho_{ABC}$ and Rényi parameters $\alpha\geq 1/2$ and $\beta\equiv\beta(\alpha)= \alpha/(2\alpha-1)$, we define
\begin{subequations}\label{eq:renyi-cmi-with-cond}
\begin{align}
\tI^{(1)}_\alpha(A;B|C)_\rho &\coloneqq \tS_\alpha(A|C)_\rho - \tS_{\beta}(A|BC)_\rho,\\
\tI^{(2)}_\alpha(A;B|C)_\rho &\coloneqq \tI_\alpha(A;BC)_\rho - \tI_{\beta}(A;C)_\rho.\label{eq:renyi-cmd-mutual}
\end{align}
\end{subequations}
These quantities satisfy the following properties:
\begin{proposition}\label{prop:renyi-gen-properties}
Let $\rho_{ABC}$ be a tripartite state and $\alpha\geq 1/2$. The quantities $\tI^{(i)}_\alpha(A;B|C)_\rho$, defined by \eqref{eq:renyi-cmi-with-cond} for $i=1,2$, satisfy:
\begin{enumerate}[{\normalfont (i)}]
\item Rényi generalization of conditional mutual information: 
\begin{align*}
\lim_{\alpha\rightarrow 1}\tI^{(i)}_\alpha(A;B|C)_\rho = I(A;B|C)_\rho.
\end{align*}

\item Monotonicity in $\alpha$: For $1/2\leq \alpha\leq\alpha'$, we have
\begin{align*}
\tI^{(1)}_\alpha(A;B|C)_\rho &\geq \tI^{(1)}_{\alpha'}(A;B|C)_\rho & \tI^{(2)}_\alpha(A;B|C)_\rho &\leq \tI^{(2)}_{\alpha'}(A;B|C)_\rho.
\end{align*}

\item Data processing inequality on the $B$ system: Let $\Lambda\colon B\rightarrow B'$ be a CPTP map and define $\sigma_{AB'C} = (\id_{AC}\ox\Lambda)(\rho_{ABC})$, then
\begin{align*}
\tI^{(i)}_\alpha(A;B|C)_\rho \geq \tI^{(i)}_\alpha(A;B'|C)_\sigma.
\end{align*}

\item Duality: Let $\rho_{ABCD}$ be a purification of $\rho_{ABC}$, then
\begin{align*}
\tI^{(1)}_\alpha(A;B|C)_\rho = \tI^{(1)}_\alpha(A;B|D)_\rho.
\end{align*}
\end{enumerate}
\end{proposition}
\begin{proof}
Property (i) follows from \cite[Lem.~8]{CMW14}.
For (ii), note that $\beta(\alpha)=\alpha/(2\alpha-1)$ is a decreasing function. 
The assertion then follows from the monotonicity in $\alpha$ of the sandwiched Rényi divergence, \Cref{prop:renyi-properties}(\ref{item:RE-mon-alpha}).\footnote{Note that the quantity $S_\beta(AB)_\rho - S_\alpha(B)_\rho$ that appears in \Cref{thm:state-re-strong-converse} is monotonic in $\alpha$ for the same reason.} 
Property (iii) is straightforward, and (iv) is obtained by employing duality (\Cref{prop:renyi-properties}(\ref{item:RE-cond-duality})).
\end{proof}

\subsection{State redistribution with feedback}\label{sec:feedback}
In this section, we consider state redistribution with feedback \cite{BCT14v2}, where the state redistribution protocol consists of $M$ rounds of forward and backward quantum communication between Alice and Bob. 
The initial state of the protocol is again the pure state $\psi_{ABCR}\ox\Phi_{T_AT_B}^k$, where systems $A$ and $C$ are with Alice, $B$ is with Bob, $R$ is an inaccessible reference system, and $\Phi_{T_AT_B}^k$ is an MES of Schmidt rank $k$ shared between Alice ($T_A$) and Bob ($T_B$). 
As before, the goal is for Alice to transfer the $A$ system to Bob, while preserving its correlations with the other systems.

The main difference with single-round state redistribution as described in \Cref{sec:state-re-protocol} is that now backward quantum communication from Bob to Alice is possible.
Furthermore, we allow for $M$ rounds of communication, in the following way (cf.~\Cref{fig:state-re-fb-protocol}): 
Alice first applies an encoding operation $\cE_1\colon ACT_A\rightarrow Q_1A_1$ to the initial state and sends $Q_1$ to Bob, who applies a decoding operation $\cD_1\colon Q_1BT_B\rightarrow Q_1'B_1$. 
The system $Q_1'$ is the quantum communication register that he sends back to Alice. 
She then applies the encoding $\cE_2\colon Q_1'A_1\rightarrow Q_2A_2$ and sends $Q_2$ to Bob, who applies the decoding $\cD_2\colon Q_2B_1\rightarrow Q_2'B_2$ and sends $Q_2'$ back, and so forth. 
In the $i$-th round, we denote by $\omega^i$ and $\sigma^i$ the states shared between Alice, Bob, and the reference, after applying the encoding $\cE_i$ and decoding $\cD_i$, respectively. 
In the final round, Alice applies the encoding $\cE_M\colon Q_{M-1}'A_{M-1}\rightarrow Q_M C' T_{A}'$ and sends $Q_M$ to Bob, who applies the decoding $\cD_M\colon Q_M B_{M-1}\rightarrow A'B'T_B'$.
The protocol succeeds if the final state is close in fidelity to the pure target state $\psi_{A'B'C'R}\ox \Phi^m_{T_A'T_B'}$, where $\psi_{A'B'C'R}=\psi_{ABCR}$ and $\Phi^m_{T_A'T_B'}$ is an MES of Schmidt rank $m$ (for some $m\in\mathbb{N}$) shared between Alice and Bob.

\begin{figure}[t]
\centering
\includegraphics[width=\textwidth]{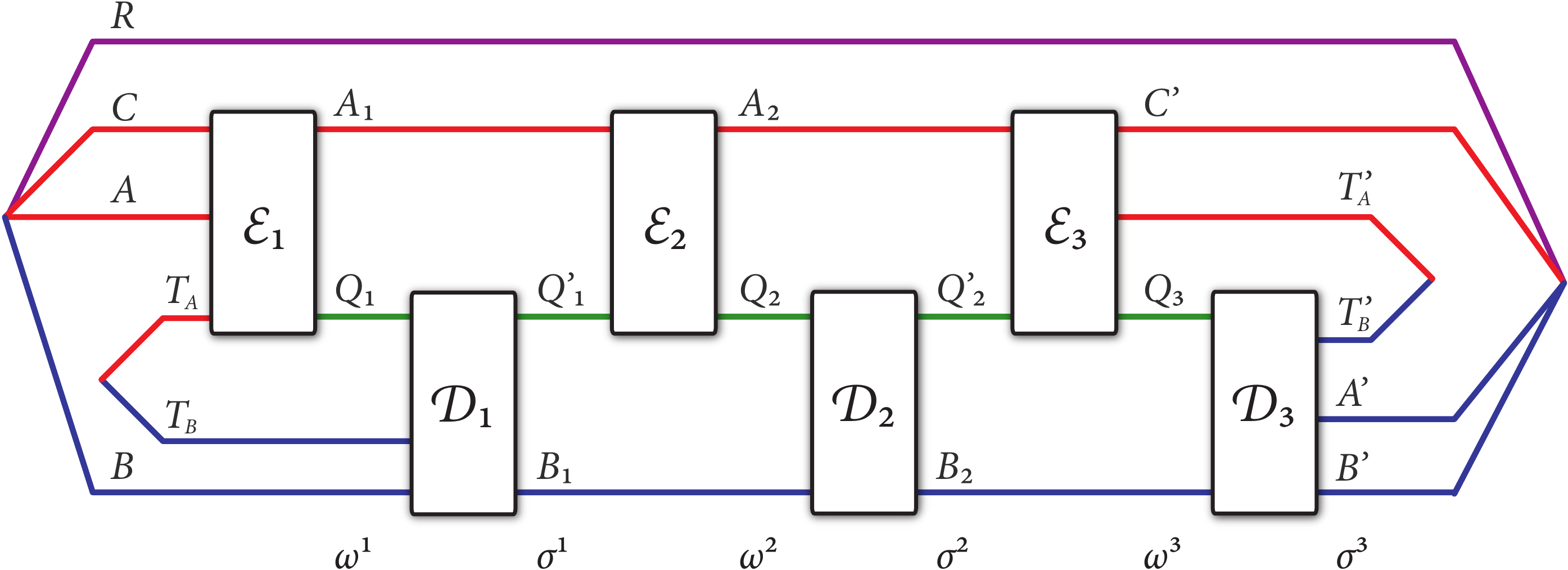}
\caption{State redistribution protocol with feedback for $M=3$ (see \Cref{sec:feedback} for a detailed description).}
\label{fig:state-re-fb-protocol}
\end{figure}

For a protocol acting on a many-copy initial state $\psi_{ABCR}^{\ox n}\ox\Phi^{k_n}_{T_AT_B}$, we define the entanglement cost $e$, the forward quantum communication $\qfw$, and the total quantum communication $\qtot$ (equal to forward plus backward communication) as
\begin{align}
e &\coloneqq \frac{1}{n} \left(\log|T_A^n| - \log |T_A'^n|\right),\label{eq:state-re-fb-ent-cost}\\
\qfw &\coloneqq \frac{1}{n} \sumi_{i=1}^M \log|Q_i^n|,\label{eq:state-re-fb-qfw}\\
\qtot &\coloneqq \qfw + \frac{1}{n} \sumi_{i=1}^{M-1} \log |Q_i'^n|.\label{eq:state-re-fb-qtot}
\end{align}
Using the smooth entropy framework, the authors in \cite{BCT14v2} proved that the achievable region for $e$, $\qfw$, and $\qtot$ is determined by the following conditions:
\begin{align}\label{eq:state-re-fb-optimal-rates}
\qfw &\geq \onehalf I(A;R|B)_\psi, & \qtot + e \geq S(A|B)_\psi.
\end{align}
We first note that both conditions are independent of $M$, the number of rounds.
Let us also compare \eqref{eq:state-re-fb-optimal-rates} to the conditions for single-round state redistribution in \eqref{eq:state-re-optimal-rates}. 
The achievable region for forward quantum communication coincides with that of \eqref{eq:state-re-optimal-rates}, as shown in \cite{BCT14v2}. 
However, for state redistribution with feedback the condition for the entanglement cost involves the \emph{total} quantum communication between the two parties. 
This can be understood to arise from the fact that quantum communication from Bob to Alice can introduce additional entanglement between them. 
In comparison to a single-round state redistribution protocol without feedback, the optimal rate of the overall entanglement cost $e$ is \emph{lowered} by the amount of backward quantum communication, since we have $\qtot\geq \qfw$.

Using the Rényi entropy method, we derive a strong converse theorem for state redistribution with feedback, \Cref{thm:state-re-fb-strong-converse} below. 
It follows from \Cref{lem:state-re-fb-sc}, which we state and prove in \Cref{sec:feedback-proofs}. 
The strong converse for state redistribution with feedback originally appeared in \cite{BCT14v2}. 
However, in contrast to \cite{BCT14v2} our proof method yields a Rényi entropic quantity as an explicit exponent in the strong converse bound. 

\begin{theorem}[Strong converse for state redistribution with feedback]\label{thm:state-re-fb-strong-converse}
Let $\Lambda_n$ denote a state redistribution protocol with feedback for a pure state $\psi_{ABCR}^{\ox n}$, as described above.
Setting
\begin{align*}
F_n \coloneqq F\left( \psi_{A'B'C'R}^{\ox n}\ox\Phi^{m_n}_{T_A'T_B'}, (\Lambda_n\ox \id_{R^n})\left(\psi_{ABCR}^{\ox n}\ox\Phi^{k_n}_{T_AT_B}\right) \right),
\end{align*}
we then have the following bounds on $F_n$ for $\alpha\in(1/2,1)$ and $\beta=\alpha/(2\alpha-1)$:
\begin{align}
F_n &\leq \exp\left\lbrace -n\kappa(\alpha) \left[S_\beta(AB)_\psi - S_\alpha(B)_\psi - (\qtot + e)\right] \right\rbrace,\\
F_n &\leq \exp\left\lbrace-n\kappa(\alpha) \left[\tS_\beta(R|B)_\psi - \tS_\alpha(R|AB)_\psi - 2\qfw \right] \right\rbrace,\\
F_n &\leq \exp\left\lbrace-n\kappa(\alpha) \left[\tI_\alpha(R;AB)_\psi - \tI_\beta(R;B)_\psi  - 2\qfw\right] \right\rbrace,
\end{align}
where $\kappa(\alpha) = (1-\alpha)/(2\alpha)$, and $e$, $\qfw$, and $\qtot$ are the entanglement cost \eqref{eq:state-re-fb-ent-cost}, forward quantum communication cost \eqref{eq:state-re-fb-qfw}, and total quantum communication cost \eqref{eq:state-re-fb-qtot}, respectively.
\end{theorem}

\subsection{Coherent state merging}\label{sec:fqsw-protocol}
Coherent state merging \cite{Opp08} is the task in which Alice wants to transfer the $A$-part of a bipartite state $\rho_{AB}$ to Bob (who holds the $B$ system), while at the same time generating entanglement between them. 
We assume that $\rho_{AB}$ is purified by an inaccessible reference system $R$. 
To achieve their goal, Alice and Bob are allowed to perform local operations on the systems in their possession as well as noiseless quantum communication.
This protocol (also known as `Fully Quantum Slepian Wolf' (FQSW) protocol \cite{ADHW09}) is a special case of the state redistribution protocol from \Cref{sec:state-re-protocol} where the system $C$ is absent (or equivalently taken to be a trivial one-dimensional system), and Alice and Bob do not share any entanglement prior to commencing the protocol (hence, the systems $T_A$ and $T_B$ are trivial and entanglement is always \emph{gained} in the course of the protocol). 

Let Alice and Bob share $n$ identical copies of the state $\rho_{AB}$ with purification $\psi_{ABR}$, i.e.~the state $\rho_{AB}^{\ox n}$ with purification $\psi_{ABR}^{\ox n}$. 
A general coherent state merging protocol $(\rho^{\ox n},\Lambda_n)$ is given by a joint quantum operation $\Lambda_n\equiv\cD_n\circ\cE_n$ where $\cE_n\colon A^n\rightarrow T_A'^nQ^n$ is Alice's encoding map, the system $Q^n$ is sent to Bob, and $\cD_n\colon Q^nB^n\rightarrow T_B'^nA'^nB'^n$ is Bob's decoding map. 
Here, $A'^n\cong A^n$, $T_B'^n\cong T_A'^n$, and $B'^n\cong B^n$. 
Denoting the final state of the protocol by $\sigma_n = (\Lambda_n\ox\id_{R^n})(\psi_{ABR}^{\ox n})$, the figure of merit is chosen to be the fidelity  
\begin{align}\label{eq:fqsw-figure-of-merit}
F_n\coloneqq F\left(\sigma_n,\psi_{ABR}^{\ox n}\ox \Phi^{m_n}_{T_A'^n T_B'^n}\right)
\end{align}
where the second argument of the fidelity is the target state of the protocol $(\rho^{\ox n},\Lambda_n)$. 

The \emph{quantum communication cost} $\Qncsm(\rho^{\ox n},\Lambda_n)$ and the \emph{entanglement gain} $\Encsm(\rho^{\ox n},\Lambda_n)$ are defined in analogy to  \Cref{sec:state-re-protocol}:
\begin{align}\label{eq:fqsw-rates}
\Qncsm(\rho^{\ox n},\Lambda_n) &\coloneqq \frac{1}{n} \log |Q^n|, & \Encsm(\rho^{\ox n},\Lambda_n) &\coloneqq \frac{1}{n} \log |T_A'^n|.
\end{align}
A pair $(e,q)$, with $e,q \geq 0$, is said to be an achievable rate pair for coherent state merging of a state $\rho_{AB}$, if there exists a sequence of protocols $\{(\rho^{\ox n},\Lambda_n)\}_{n \in \mathbb{N}}$ such that $ \liminf_{n \to \infty} F_n = 1$ and 
\begin{align*}
\limsup_{n \to \infty} \Qncsm(\rho^{\ox n}, \Lambda_n) &= q, & \liminf_{n \to \infty} \Encsm(\rho^{\ox n}, \Lambda_n) &= e.
\end{align*}

Coherent state merging was introduced by Abeyesinghe \etal \cite{ADHW09} and further investigated in \cite{DH11,BCR11}. 
It was proved that a rate pair $(e,q)$ is achievable if and only if $e$ and $q$ satisfy the conditions
\begin{align*}
q &\geq \onehalf I(A;R)_\rho, & q-e &\geq S(A|B)_\rho.
\end{align*}

As mentioned above, every coherent state merging protocol can be seen as a special case of a state redistribution protocol where the systems $C$ and $T_A$ are trivial. 
In this case, $I(A;R|B)_\rho=I(A;R)_\rho$, and \Cref{lem:state-re-converse} reduces to

\begin{lemma}\label{lem:fqsw-converse-one-shot}
Let $\rho\equiv\rho_{AB}$ be a bipartite state with purification $|\psi_{ABR}\rangle$, and let $(\rho,\Lambda)$ be a coherent state merging protocol where $\Lambda\equiv \cD\circ\cE$ with $\cE\colon A\rightarrow T_A'Q$ and $\cD\colon QB\rightarrow T_B'A'B'$ as defined above (for $n=1$). 
Furthermore, set
\begin{align*} 
F \coloneqq F(\Phi_{T_A'T_B'}\ox\psi_{B'BR},(\cD\circ\cE\ox\id_R)(\psi_{ABR})).
\end{align*} 
Then we have the following bounds for $\alpha\in(1/2,1)$ and $\beta=\alpha/(2\alpha-1)$:
\begin{align}
\log F &\leq \frac{1-\alpha}{2\alpha} \left( \log|Q| - \log|T_A'| - S_\beta(R)_\rho + S_\alpha(AR)_\rho \right),\\ 
\log F &\leq \frac{1-\alpha}{2\alpha} ( 2\log|Q| - S_\beta(R)_\rho + \tS_\alpha(R|A)_\rho). 
\end{align}
\end{lemma}

\Cref{lem:fqsw-converse-one-shot} in turn implies the strong converse property for the quantum communication cost and entanglement gain of coherent state merging:
\begin{theorem}[Strong converse for coherent state merging]\label{thm:fqsw-strong-converse}
Let $\rho\equiv\rho_{AB}$ be a bipartite state and $\lbrace(\rho^{\ox n},\Lambda_n)\rbrace_{n\in\mathbb{N}}$ be a sequence of coherent state merging protocols as described above, with figure of merit $F_n$ as defined in \eqref{eq:fqsw-figure-of-merit}.
Then for all $n\in\mathbb{N}$ we obtain the following bounds on the fidelity $F_n$ for $\alpha\in(1/2,1)$ and $\beta=\alpha/(2\alpha-1)$:
\begin{align} 
F_n &\leq \exp\left\lbrace - n\kappa(\alpha) \left[S_\beta(AB)_\rho - S_\alpha(B)_\rho - \Qncsm + \Encsm\right]\right\rbrace,\\
F_n &\leq \exp\left\lbrace - n\kappa(\alpha) \left[S_\beta(R)_\rho - \tS_\alpha(R|A)_\rho- 2\Qncsm \right] \right\rbrace.
\end{align}
where $\kappa(\alpha)=(1-\alpha)/(2\alpha)$, and $\Qncsm\equiv\Qncsm(\rho^{\ox n},\Lambda_n)$ and $\Encsm\equiv \Encsm(\rho^{\ox n},\Lambda_n)$ denote the quantum communication cost and entanglement gain respectively, as defined in \eqref{eq:fqsw-rates}.
\end{theorem}

\subsection{Quantum state splitting}\label{sec:state-splitting}
Quantum state splitting is the task in which a tripartite pure state $\psi_{ACR}$, which is initially shared between Alice (who has $AC$) and the reference ($R$), is split between Alice, Bob and the reference, with the system $A$ being transferred to Bob. 
To this end, Alice and Bob can make use of prior shared entanglement and are allowed to do local operations on systems which they possess or receive.
This protocol (also known as `Fully Quantum Reverse Shannon' (FQRS) protocol \cite{ADHW09,Dev06}) is dual to the coherent state merging protocol from \Cref{sec:fqsw-protocol} under time reversal \cite{Dev06}. 
Hence, the quantum state splitting protocol can also be obtained as a special case from the state redistribution protocol, if the systems $B$ and $T_A'$ are taken to be trivial. 
That is, Bob does not possess a share of the input state of the protocol, and the target state does not consist of an MES shared between Alice and Bob (i.e.~the protocol always \emph{consumes} entanglement).

Let Alice and Bob share $n$ identical copies of the state $\rho_{AC}$ with purification $\psi_{ACR}$, i.e.~the state $\rho_{AC}^{\ox n}$ with purification $\psi_{ACR}^{\ox n}$. 
A general quantum state splitting protocol $(\rho^{\ox n},\Lambda_n)$ is given by a joint quantum operation $\Lambda_n=\cD_n\circ\cE_n$ where $\cE_n\colon A^nC^nT_A^n\rightarrow C'^nQ^n$ is Alice's encoding map, the system $Q^n$ is sent to Bob, and $\cD_n\colon Q^nT_B^n\rightarrow A'^n$ is Bob's decoding map. 
Here, $A'^n\cong A^n$ and $C'^n\cong C^n$. 
Denote the final state of the protocol by $\sigma_n = (\Lambda_n\ox\id_{R^n})(\Omega^n)$ where $\Omega^n\equiv \psi_{ACR}^{\ox n}\ox\Phi^{k_n}_{T_A^n T_B^n}$ is the initial state shared between Alice and Bob. 
Then the figure of merit is chosen to be the fidelity 
\begin{align}\label{eq:fqrs-figure-of-merit}
F_n\coloneqq F\left(\sigma_n,\psi^{\ox n}\right),
\end{align}
where $\psi\equiv\psi_{A'C'R}$.

The \emph{quantum communication cost} $\Qnqss(\rho^{\ox n},\Lambda_n)$ and the \emph{entanglement cost} $\Enqss(\rho^{\ox n},\Lambda_n)$ are defined in analogy to  \Cref{sec:state-re-protocol}:
\begin{align}\label{eq:qss-rates}
\Qnqss(\rho^{\ox n},\Lambda_n) &\coloneqq \frac{1}{n} \log |Q^n|, & \Enqss(\rho^{\ox n},\Lambda_n) &\coloneqq \frac{1}{n} \log |T_A^n|.
\end{align}
A pair $(e,q)$, with $e,q \geq 0$, is said to be an achievable rate pair for quantum state splitting of a state $\rho_{AC}$, if there exists a sequence of protocols $\{(\rho^{\ox n},\Lambda_n)\}_{n \in \mathbb{N}}$ such that $ \liminf_{n \to \infty} F_n = 1$ and 
\begin{align*}
\limsup_{n \to \infty} \Enqss(\rho^{\ox n}, \Lambda_n)&=e, & \limsup_{n \to \infty} \Qnqss(\rho^{\ox n}, \Lambda_n) &= q.
\end{align*}

The optimal rates of entanglement cost and quantum communication cost for quantum state splitting were investigated in \cite{ADHW09,Dev06,BDHSW14,BCR11}: 
A rate pair $(e,q)$ is achievable if and only if $e$ and $q$ satisfy
\begin{align*}
q &\geq \onehalf I(A;R)_\rho, & q+e &\geq S(A)_\rho.
\end{align*}
One-shot bounds characterizing the quantum communication cost and entanglement cost for quantum state splitting were derived by Berta \etal \cite{BCR11} as a building block in a proof of the Quantum Reverse Shannon theorem based on smooth entropies. 

As mentioned at the beginning of this section, quantum state splitting is a special case of state redistribution with the choices $|B|=|T_A'|=1$.
In this case, $I(A;R|B)_\rho=I(A;R)$, and \Cref{lem:state-re-converse} reduces to

\begin{lemma}\label{lem:state-splitting-converse-one-shot}
Let $\rho\equiv\rho_{AB}$ be a bipartite state with purification $|\psi_{ABR}\rangle$, and let $(\rho,\Lambda)$ be a quantum state splitting protocol where $\Lambda\equiv \cD\circ\cE$ with $\cE\colon AA'T_A\rightarrow AQ$ and $\cD\colon QT_B\rightarrow B$ as defined above (for $n=1$). 
Furthermore, set
\begin{align*} 
F \coloneqq F(\psi_{ABR},(\cD\circ\cE\ox\id_R)(\psi_{AA'R}\ox\Phi_{T_AT_B})).
\end{align*} 
Then we have the following bounds for $\alpha\in(1/2,1)$ and $\beta=\alpha/(2\alpha-1)$:
\begin{align}
\log F &\leq \frac{1-\alpha}{2\alpha} \left( \log|Q| + \log|T_A| - S_\beta(A)_\rho \right), \\
\log F &\leq \frac{1-\alpha}{2\alpha} \left( 2\log|Q| - S_\beta(R)_\rho + \tS_\alpha(R|A)_\rho \right),\\ 
\log F & \leq \frac{1-\alpha}{2\alpha} \left( 2\log|Q| - \tI_{\alpha}(R;A)_\rho\right).
\end{align}
\end{lemma}

\Cref{lem:state-splitting-converse-one-shot} now implies a strong converse theorem for state splitting:

\begin{theorem}[Strong converse for quantum state splitting]\label{thm:qss-strong-converse}
Let $\rho\equiv\rho_{AB}$ be a bipartite state and $\lbrace(\rho^{\ox n},\Lambda_n)\rbrace_{n\in\mathbb{N}}$ be a sequence of quantum state splitting protocols as described above, with figure of merit $F_n$ as defined in \eqref{eq:fqrs-figure-of-merit}.
Then for all $n\in\mathbb{N}$ we obtain the following bounds on the fidelity $F_n$ for $\alpha\in(1/2,1)$ and $\beta=\alpha/(2\alpha-1)$:
\begin{align}
F_n &\leq \exp\left\lbrace - n\kappa(\alpha) \left[S_\beta(A)_\rho - (\Qnqss + \Enqss) \right] \right\rbrace,\\
F_n &\leq \exp\left\lbrace - n\kappa(\alpha) \left[S_\beta(R)_\rho - \tS_\alpha(R|A)_\rho - 2\Qnqss\right]\right\rbrace,\\
F_n &\leq \exp\left\lbrace - n\kappa(\alpha) \left[\tI_\alpha(R;A)_\rho - 2\Qnqss\right]\right\rbrace,
\end{align}
where $\kappa(\alpha)=(1-\alpha)/(2\alpha)$, and $\Qnqss\equiv\Qnqss(\rho^{\ox n},\Lambda_n)$ and $\Enqss\equiv\Enqss(\rho^{\ox n},\Lambda_n)$ denote the quantum communication cost and entanglement cost defined in \eqref{eq:qss-rates}, respectively.
\end{theorem}

\section{Measurement compression with quantum side information}\label{sec:measurement-comp}
\subsection{The protocol}\label{sec:meas-comp-protocol}
Consider a bipartite state $\rho_{AB}$ between two parties (say, Alice and Bob), and a positive operator-valued measure (POVM) $\Lambda = \lbrace \Lambda_x\rbrace_{x\in\cX}$ (i.e.~$0\leq \Lambda_x\leq \one_A$ for all $x\in\cX$ and $\sum_{x\in\cX}\Lambda_x=\one_A$) on the $A$ system, where $X$ denotes a classical register.
Suppose that Alice wants to communicate the measurement outcome $X$ of $\Lambda$ to Bob via classical communication.
A simple solution is of course for Alice to apply the POVM $\Lambda$ and send the outcome to Bob, requiring $\log|X|$ bits of communication.
In measurement compression with quantum side information, Alice and Bob want to reduce this communication cost by \emph{simulating} the POVM $\Lambda$ using shared randomness, Bob's quantum side information $B$, and sending $\log|L|$ bits of classical communication with $|L|\leq|X|$.
This information-theoretic task was introduced in \cite{WHBH12} as an extension of Winter's original formulation of measurement compression \cite{Win04}.
In the following, we explain this protocol in more detail.

Given $\rho_{AB}\in\cD(\cH_{AB})$ and a POVM $\Lambda = \lbrace \Lambda_x\rbrace_{x\in\cX}$ on the $A$ system with outcome $X$, a general protocol for measurement compression with quantum side information consists of the following steps (see also \Cref{fig:meas-comp}): 
Alice applies a quantum operation $\cE\colon AM_A \rightarrow \bX L$ to her shares of a purification $\psi_{RAB}$ of the initial state $\rho_{AB}$ (with $R$ being an inaccessible reference system) and the shared randomness $\chi_{M_AM_B}$.
This produces a classical register $\bX$ that holds her copy of the simulated outcome of the measurement, and a classical register $L$. 
She sends the latter to Bob, who then applies a quantum operation $\cD\colon L BM_B \rightarrow \hX B'$ to $L$ and his shares of $\psi_{RAB}$ and $\chi_{M_AM_B}$, producing a quantum output $B'$ and the classical output $\hX$, which represents the simulated outcome of the measurement.
We denote the overall state of the protocol after applying $\cE$ and $\cD$ by $\omega$ and $\sigma$, respectively (cf.~\Cref{fig:meas-comp}). 
In comparison, we consider the ideal state $\varphi_{RXX'B}$ that would result from Alice applying the POVM $\Lambda$ to her system $A$ yielding the outcome $X$, and sending a copy $X'$ uncompressed to Bob.
That is, for $\zeta_A\in\cD(\cH_A)$ we define the measurement channel $\Lambda^{A\to XX'}(\zeta_A)\coloneqq \sum_{x\in\cX} \tr(\Lambda_x \zeta_A) |x\rangle\langle x|_X \ox |x\rangle\langle x|_{X'}$ associated to the POVM $\Lambda$, and set $\varphi_{RXX'B} \coloneqq (\id_{RB}\ox\Lambda^{A\to XX'})(\psi_{RAB})$.
The aim of the measurement compression protocol is to assure that $\sigma$ is close in fidelity to the ideal state $\varphi$:
\begin{align}\label{eq:meas-comp-fidelity}
F\coloneqq F\left(\varphi_{RXX'B},\sigma_{R\bX \hX B'}\right) = F\left(\varphi_{RXX'B},(\id_R\ox \cD\circ\cE)(\psi_{RAB}\ox\chi_{M_AM_B})\right).
\end{align}

\begin{figure}[t]
\centering
\includegraphics[width=0.7\textwidth]{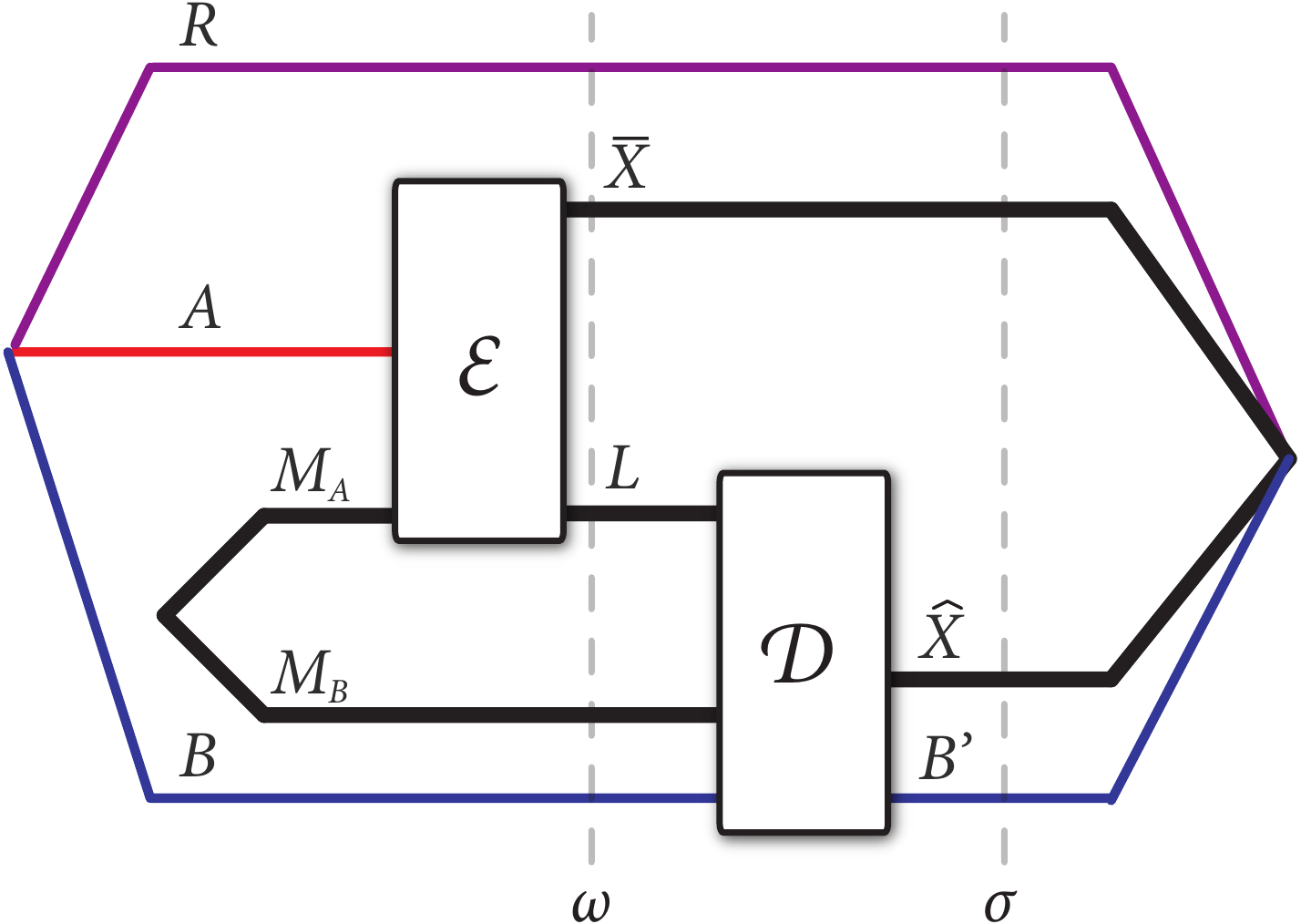}
\caption{Measurement compression with quantum side information (see \Cref{sec:meas-comp-protocol} for a detailed explanation).}
\label{fig:meas-comp}
\end{figure}

Given $n$ identical copies of the input state $\rho_{AB}$ with purification $\psi_{RAB}$ and the shared randomness $\chi_{M_A^n M_B^n}$, we consider a measurement compression protocol with maps $\cE_n\colon A^nM_A^n \rightarrow \bX^n L^n$ and $\cD_n\colon L^nB^nM_B^n\rightarrow \hX^n B'^n$, where $L^n$ is the classical communication between Alice and Bob.
The figure of merit is then given by
\begin{align}\label{eq:meas-comp-fidelity-iid}
F_n\coloneqq  F\left(\vphi_n, (\id_{R^n}\ox \cD_n\circ\cE_n)\left(\psi_{RAB}^{\ox n}\ox \chi_{M_A^n M_B^n}\right)\right)
\end{align}
where the ideal state $\vphi_n$ is obtained by Alice applying the POVM $\Lambda^{\ox n}$ yielding the outcome $X^n$ and sending a copy $X'^n$ to Bob. We define the \emph{classical communication cost} 
\begin{align}\label{eq:cc-cost}
c\left(\rho^{\ox n}, \Lambda^{\ox n}\right) \coloneqq \frac{1}{n} \log |L^n|,
\end{align}
and the \emph{randomness cost}
\begin{align}\label{eq:rand-cost}
r\left(\rho^{\ox n}, \Lambda^{\ox n}\right) \coloneqq \frac{1}{n} \log |M_A^n|.
\end{align}
A rate pair $(c,r)$ with $c,r\geq 0$ is called achievable if there exists a sequence $\lbrace(\cE_n,\cD_n)\rbrace_{n\in\mathbb{N}}$ of measurement compression protocols such that
\begin{align*}
\liminf_{n\to\infty} F_n &= 1, & \limsup_{n\to\infty} c\left(\rho^{\ox n}, \Lambda^{\ox n}\right) &= c, & \limsup_{n\to\infty} r\left(\rho^{\ox n}, \Lambda^{\ox n}\right) &= r.
\end{align*}
In \cite{WHBH12} it was proved that $(c,r)$ is achievable if and only if 
\begin{subequations}\label{eq:meas-comp-optimal}
\begin{align}
c &\geq I(X;R|B)_\vphi\label{eq:meas-comp-optimal-c},\\
c + r &\geq S(X|B)_\vphi,\label{eq:meas-comp-optimal-c+r}
\end{align}
\end{subequations}
where $\vphi_{RXX'B}$ is the ideal state of the protocol defined above.

\subsection{Strong converse theorem}
In this section, we strengthen the weak converse result obtained from \eqref{eq:meas-comp-optimal-c} for the classical communication cost in measurement compression with quantum side information to a strong converse theorem. As in \Cref{sec:state-re-strong-converse}, we first derive the following `one-shot' lemma:
\begin{lemma}\label{lem:meas-comp}
Let a bipartite state $\rho_{AB}$ with purification $\psi_{RAB}$ and a POVM $\Lambda$ on $A$ be given. Furthermore, let $\lbrace( \cE,\cD)\rbrace$ be a measurement compression protocol as defined in \Cref{sec:meas-comp-protocol} with figure of merit
\begin{align*}
F\coloneqq  F\left(\varphi_{RXX'B},(\id_R\ox \cD\circ\cE)(\psi_{RAB}\ox\chi_{M_AM_B})\right).
\end{align*}
Then we have the following bound on $F$ for $\alpha\in (1/2,1)$ and $\beta\equiv \beta(\alpha) = \alpha/(2\alpha-1)$:
\begin{align}
\log F &\leq \frac{1-\alpha}{2\alpha} \left( \log|L| - \tS_\beta( R|B)_\varphi + \tS_\alpha(R|XB)_\vphi \right).\label{eq:meas-comp-c} 
\end{align}
\end{lemma}

\begin{proof}
Define the states 
\begin{align*}
\omega_{R\bX L BM_B} &\coloneqq (\id_R\ox \cE)(\psi_{RAB}\ox\chi_{M_AM_B}),\\
\sigma_{R\bX\hX B'} &\coloneqq (\id_R\ox\cD\circ\cE)(\psi_{RAB}\ox\chi_{M_AM_B}).
\end{align*} 
To prove \eqref{eq:meas-comp-c}, consider the following bound for $\alpha\in (1/2,1)$ and $\beta = \alpha/(2\alpha-1)$:
\begin{align}
\frac{2\alpha}{1-\alpha}\log F &\leq \frac{2\alpha}{1-\alpha}\log F\left(\sigma_{R \hX B'},\vphi_{RXB}\right)\notag,\\
&\leq \tS_\alpha(R|XB)_\vphi - \tS_\beta(R|\hX B')_\sigma\label{eq:mmnt-comp-fidelity-bound}
\end{align}
where the first inequality follows from the monotonicity of the fidelity under partial trace, and the second inequality follows from \cref{eq:renyi-fidelity-conditional} of \Cref{prop:renyi-fidelity}.
We continue to bound the second term on the right-hand side of \eqref{eq:mmnt-comp-fidelity-bound}:
\begin{align*}
-\tS_\beta(R|\hX B')_\sigma &\leq - \tS_\beta(R|LB M_B)_\omega\\
&\leq \log |L| - \tS_\beta(R|BM_B )_\omega\\
&= \log |L| - \tS_\beta(R|BM_B )_{\psi\ox\chi}\\
&= \log |L| - \tS_\beta(R|B)_\vphi\numberthis.\label{eq:mmnt-comp-second-term}
\end{align*}
The first inequality follows from data processing with respect to the quantum operation $\cD\colon LM_BB\rightarrow \hX B'$ (\Cref{prop:renyi-properties}(\ref{item:dpi})), the second inequality follows from \cref{eq:classical-dim-bound} of \Cref{prop:renyi-cq-states}, and the first equality follows from the fact that $\omega_{RB M_B} = \psi_{RB}\ox\chi_{M_B}$. The last equality uses \cref{eq:product-conditional} of \Cref{lem:renyi-quantities}, and the fact that $\psi_{RB}=\vphi_{RB}$.
Substituting \eqref{eq:mmnt-comp-second-term} in \eqref{eq:mmnt-comp-fidelity-bound} then yields the claim. 
\end{proof}

This immediately implies the following strong converse theorem:
\begin{theorem}[Strong converse theorem for measurement compression with QSI]\label{thm:meas-comp-strong-converse}
Let $\rho_{AB}$ be a bipartite state, $\Lambda$ a given POVM on $A$, and $\lbrace (\cE_n,\cD_n)\rbrace_{n\in\mathbb{N}}$ be a sequence of measurement compression protocols as described in \Cref{sec:meas-comp-protocol}, with figure of merit $F_n$ as defined in \eqref{eq:meas-comp-fidelity-iid}. Then for all $n\in\mathbb{N}$ we have the following bound on $F_n$ for $\alpha\in(1/2,1)$ and $\beta=\alpha/(2\alpha-1)$:
\begin{align}
F_n \leq \exp\left\lbrace - n\kappa(\alpha)\left[\tS_\beta(R|B)_\vphi - \tS_\alpha(R|XB)_\vphi - c\right]\right\rbrace,
\end{align}
where $\kappa(\alpha)=(1-\alpha)/(2\alpha)$ and $c\equiv c(\rho^{\ox n},\Lambda^{\ox n})$ is the classical communication cost defined in \eqref{eq:cc-cost}.
\end{theorem}
\begin{remark}
The achievable rate region in the $(c,r)$-plane is determined by the two boundaries $c\geq I(X;R|B)_\vphi$ and $c+r\geq S(X|B)_\vphi$, as stated in \eqref{eq:meas-comp-optimal} (compare this to the similar situation in the state redistribution protocol discussed in \Cref{sec:state-re}). 
\Cref{thm:meas-comp-strong-converse} only proves the strong converse property for the $c$-boundary of the achievable rate region, and it remains open to prove the strong converse property also for the $(c+r)$-boundary as stated in \eqref{eq:meas-comp-optimal-c+r}.
While the proof of \eqref{eq:meas-comp-c} in \Cref{lem:meas-comp} closely follows that of \eqref{eq:state-re-q} in \Cref{lem:state-re-converse}, our investigations indicate that the proof method of \eqref{eq:state-re-q+e} does not immediately carry over to show the desired bound for $c+r$ in measurement compression, that is,
\begin{align}
\log F \overset{?}{\leq} f(\alpha)(\log|L| + \log|M_A| - S_{\beta(\alpha)}(XB)_\vphi + S_\alpha(B)_\vphi),
\end{align}
for some functions $f(\alpha)$ and $\beta(\alpha)$ satisfying $f(\alpha)>0$ for all $\alpha$ in some open interval whose boundary contains $1$, and $\lim_{\alpha\to 1}\beta(\alpha) = 1$.
\end{remark}

\section{Randomness extraction}\label{sec:randex}

\subsection{The protocol}\label{sec:randex-protocol}
Suppose that Alice and Bob share the c-q state $\rho_{XB} = \sum_{x\in\cX} p_x |x\rangle\langle x|_X \ox \rho_B^x$ with $\rho_B^x\in\cD(\cH_B)$ for all $x\in\cX$, where the classical register $X$ is with Alice and the quantum system $B$ is with Bob. 
The goal of a randomness extraction protocol is to extract from $X$ a random string $Z$ that is uncorrelated with $B$.

In the asymptotic, memoryless setting Alice and Bob share $n$ copies of the c-q state $\rho_{XB}$. A general randomness extraction protocol $(\rho_{XB}^{\ox n},e_n)$ consists of a (surjective) encoding function $e_n\colon \cX^n\rightarrow \cZ^n$ where $\cX^n = \cX^{\times n}$. The classical register $Z^n$ is then defined to be the one associated with the set $\cZ^n$. 
The encoding function $e_n$ gives rise to an encoding (quantum) operation, and without loss of generality this encoding map can be taken to be an isometry
$U_{e_n}\colon |x^n\rangle \mapsto |x^n\rangle \ox |e_n(x^n)\rangle$ where $|e_n(x^n)\rangle\in\lbrace|z^n\rangle\rbrace_{z^n\in\cZ^n}$ for all $x^n\in\cX^n$, resulting in the encoded state 
\begin{align}\label{eq:randex-omega}
\omega_{X^nZ^nB^n} \coloneqq U_{e_n} \rho_{XB}^{\ox n} U_{e_n}^\dagger = \sum_{x^n\in\cX^n} p_{x^n} |x^n\rangle\langle x^n|_{X^n} \ox |e_n(x^n)\rangle\langle e_n(x^n)|_{Z^n}\ox\rho_{B^n}^{x^n},
\end{align}
with $\rho_{B^n}^{x^n}\coloneqq \rho_B^{x_1}\otimes \dots\otimes \rho_B^{x_n}$ for $x^n=x_1\dots x_n\in\cX^n$.
Upon discarding $X^n$, the final state of the protocol is then given by 
\begin{align}\label{eq:randex-final-state}
\omega_n\equiv \omega_{Z^nB^n} = \sum_{z^n\in\cZ^n} |z^n\rangle \langle z^n|_{Z^n} \ox \sum_{x^n\in e_n^{-1}(z^n)} p_{x^n}\rho_{B^n}^{x^n}.
\end{align}
The randomness extraction protocol $(\rho_{XB}^{\ox n},e_n)$ succeeds if the final state $\omega_n$ is close to a state that is completely mixed on $Z^n$ and independent of $B^n$. 
As the figure of merit we choose the fidelity 
\begin{align}\label{eq:randex-figure-of-merit}
F_n\coloneqq \max_{\sigma_{B^n}}F\left(\omega_n,\pi_Z^{\ox n}\ox \sigma_{B^n}\right).
\end{align}
The rate of extractable randomness $l(\rho_{XB}^{\ox n},e_n)$ is defined as
\begin{align}\label{eq:extractable-randomness}
l(\rho_{XB}^{\ox n},e_n) \coloneqq \frac{1}{n}\log |Z^n|.
\end{align}
A real number $l\geq 0$ is said to be an achievable rate for randomness extraction if there is a sequence of protocols $\lbrace (\rho_{XB}^{\ox n},e_n)\rbrace_{n\in\mathbb{N}}$ such that 
\begin{align*}
\liminf_{n\to \infty} F_n &=1, & \liminf_{n\to \infty}l\left(\rho_{XB}^{\ox n},e_n\right) &= l.
\end{align*}

Randomness extraction was first studied by Bennett \etal \cite{BBR88} (under the name of `privacy amplification') and further developed by Renner \cite{Ren05} and Renner and König \cite{RK05} (see also \cite{TSSR11}). 
They showed that $l\geq 0$ is an achievable rate for randomness extraction if and only if
\begin{align}\label{eq:randex-optimal-rate}
l \leq S(X|B)_\rho.
\end{align}
Tomamichel \cite{Tom12} proved a strong converse theorem for randomness extraction based on one-shot bounds in terms of smooth entropies.

\subsection{Strong converse theorem}\label{sec:randex-strong-converse}
We first state the following general bound on the fidelity: 

\begin{lemma}[\cite{Sha14}]\label{lem:fidelity-bound}
Let $\rho_{AB}\in\cD(\cH_{AB})$, $\sigma_{A}\in\cD(\cH_A)$, and $\chi_B\in\cD(\cH_B)$ be arbitrary quantum states, then
\begin{align*}
F(\rho_{AB},\sigma_A\ox\rho_B) \geq F^2(\rho_{AB},\sigma_A\ox\chi_B).
\end{align*}
\end{lemma}

With this result in hand, we can prove the following bound on the fidelity for arbitrary randomness extraction protocols:

\begin{lemma}\label{lem:randex-converse}
Let $\rho_{XB}=\sum_{x\in\cX}p_x |x\rangle\langle x|_X\ox \rho_B^x$ be a c-q state with $\rho_B^x\in\cD(\cH_B)$ for all $x\in\cX$, and denote by $\omega_{XZB}$ the encoded state of a randomness extraction protocol $(\rho_{XB},e)$ as defined in \Cref{sec:randex-protocol} (for $n=1$). 
Furthermore, set
\begin{align*}
F\coloneqq \max_{\sigma_B} F(\omega_{ZB},\pi_Z\ox \sigma_B).
\end{align*}
Then we have the following bound for all $\alpha\in (1/2,1)$ and $\beta\equiv\beta(\alpha)=\alpha/(2\alpha-1)$:
\begin{align}
\log F &\leq \frac{1-\alpha}{4\alpha} \left(S_\alpha(XB)_\rho - S_\beta(B)_\rho - \log|Z|\right).\label{eq:randex-fidelity-linear-renyi}
\intertext{For the same range of $\alpha$, we also obtain the following alternative bound on $F$:}
\log F &\leq \frac{1-\alpha}{4\alpha} \left(\tS_\alpha(X|B)_\rho - \log|Z|\right).\label{eq:randex-fidelity-renyi-cond}
\end{align}
\end{lemma}

\begin{proof}
We first prove \eqref{eq:randex-fidelity-linear-renyi}. 
Set $F'\coloneqq F(\omega_{ZB},\pi_Z\ox\rho_B)$, then by \cref{eq:renyi-fidelity-entropy} of \Cref{prop:renyi-fidelity} we have the following bound for $\alpha\in(1/2,1)$ and $\beta=\alpha/(2\alpha-1)$:
\begin{align}
S_\alpha(ZB)_\omega &\geq S_\beta(ZB)_{\pi\ox\rho} + \frac{2\alpha}{1-\alpha} \log F' \nonumber\\
&= \log|Z| + S_\beta(B)_\rho + \frac{2\alpha}{1-\alpha} \log F'.\label{eq:log-Z-fidelity}
\end{align}
In the second line we used additivity of the Rényi entropy, as well as the fact that $S_\gamma(\pi_Z) = \log|Z|$ for all $\gamma\geq0$ (cf.~\Cref{prop:renyi-properties}(\ref{item:RE-add}) and (\ref{item:RE-dim-bound})).

Furthermore, we have the bound
\begin{align}
S_\alpha(XB)_\rho = S_\alpha(XZB)_\omega \geq S_\alpha(ZB)_\omega,\label{eq:XZB-bound}
\end{align}
where the equality follows from the invariance of the Rényi entropy under the encoding isometry $U_e$ (\Cref{prop:renyi-properties}(\ref{item:RE-isom})), and the inequality follows from \cref{eq:discarding-classical-info} of \Cref{prop:renyi-cq-states}.

Putting \eqref{eq:log-Z-fidelity} and \eqref{eq:XZB-bound} together, we obtain
\begin{align}\label{eq:F-prime-bound}
\log F'\leq \frac{1-\alpha}{2\alpha}\left(S_\alpha(XB)_\rho - S_\beta(B)_\rho - \log |Z|\right).
\end{align}
Now observe that $\omega_B=\rho_B$, and by \Cref{lem:fidelity-bound} we have 
\begin{align}\label{eq:rand-ex-fidelity-relation}
F' = F(\omega_{ZB},\pi_Z\ox\rho_B) = F(\omega_{ZB},\pi_Z\ox\omega_B) \geq F^2(\omega_{ZB},\pi_Z\ox \sigma_B)
\end{align}
for all $\sigma_B\in\cD(\cH_B)$. 
Substituting this into \eqref{eq:F-prime-bound} and using the monotonicity of the logarithm then yields the claim.

To prove \eqref{eq:randex-fidelity-renyi-cond}, observe that \eqref{eq:rand-ex-fidelity-relation} together with \cref{eq:renyi-fidelity-conditional} of \Cref{prop:renyi-fidelity} yield the following for $\alpha\in(1/2,1)$ and $\beta=\alpha/(2\alpha-1)$:
\begin{align*}
\frac{4\alpha}{1-\alpha}\log F &\leq \frac{2\alpha}{1-\alpha}\log F'\\
&\leq \tS_\alpha(Z|B)_\omega - \tS_\beta(Z|B)_{\pi\ox\rho}\\
&= \tS_\alpha(Z|B)_\omega - \log|Z|\\
&\leq \tS_\alpha(XZ|B)_\omega - \log|Z|\\
&= \tS_\alpha(X|B)_\rho - \log|Z|,
\end{align*}
where the first equality follows from \cref{eq:product-conditional} of \Cref{lem:renyi-quantities}, the third inequality uses \cref{eq:discarding-classical-info} of \Cref{prop:renyi-cq-states}, and the last equality uses the invariance of the Rényi conditional entropy under the isometry $U_e$.
\end{proof}

This implies the following strong converse theorem for randomness extraction:
\begin{theorem}\label{thm:randex-strong-converse}
Let $\rho_{XB}$ be a c-q state, and let $\lbrace(\rho_{XB}^{\ox n},e_n)\rbrace_{n\in\mathbb{N}}$ be a sequence of randomness extraction protocols as defined in \Cref{sec:randex-protocol} with figure of merit $F_n$ as given by \eqref{eq:randex-figure-of-merit}. 
Then for all $n\in\mathbb{N}$ we have the following bounds on $F_n$ for $\alpha\in(1/2,1)$ and $\beta=\alpha/(2\alpha-1)$:
\begin{align}
F_n &\leq \exp\left\lbrace-\kappa(\alpha) n\left[l - S_\alpha(XB)_\rho + S_\beta(B)_\rho \right]\right\rbrace,\label{eq:randex-sc-linear}\\
F_n &\leq \exp\left\lbrace-\kappa(\alpha) n\left[l - \tS_\alpha(X|B)_\rho \right]\right\rbrace,\label{eq:randex-sc-cond}
\end{align}
where $\kappa(\alpha)=(1-\alpha)/(4\alpha)$, and $l\equiv l(\rho_{XB}^{\ox n},e_n)$ is the rate of extractable randomness defined in \eqref{eq:extractable-randomness}.
\end{theorem}

\begin{remark}\label{rem:bound-comparison}
An immediate question arising from \Cref{thm:randex-strong-converse} is whether one of the two bounds in \eqref{eq:randex-sc-linear} and \eqref{eq:randex-sc-cond} is tighter than the other, that is, whether one of 
\begin{align}\label{eq:bound-comparison}
S_\alpha(XB)_\rho - S_\beta(B)_\rho &\overset{?}{\leq} \tS_\alpha(X|B)_\rho, & S_\alpha(XB)_\rho - S_\beta(B)_\rho &\overset{?}{\geq} \tS_\alpha(X|B)_\rho
\end{align} 
holds for all $\alpha\in(1/2,1)$ and $\beta=\alpha/(2\alpha-1)$. 
Numerical investigations with classical registers $B$ show that neither inequality in \eqref{eq:bound-comparison} is always valid. 
Hence, the two exponents in \eqref{eq:randex-sc-linear} and \eqref{eq:randex-sc-cond} are in general incomparable. 
In the light of identifying the strong converse exponent of randomness extraction against quantum side information (cf.~\Cref{sec:discussion}), this fact indicates that further analysis is needed here, as a strong converse exponent usually characterizes the tightest possible strong converse bound.
Interchanging $\alpha$ and $\beta$ in the above, the same reasoning applies to \Cref{thm:dc-strong-converse} of the following section.
\end{remark}

\section{Data compression with quantum side information}\label{sec:dc}

\subsection{The protocol}\label{sec:dc-protocol}
In data compression with quantum side information, Alice has a classical register $X$, and Bob holds the quantum system $B$ (the `side information') which is correlated with $X$. 
The goal of the protocol is for Alice to encode her classical register $X$ in a (smaller) classical register $C$ such that Bob can recover $X$ from $C$ and his quantum system $B$.
A data compression protocol $(\rho_{XB},e,\cD)$ can be described by the following steps:

The initial state is a classical-quantum state $\rho_{XB} = \sum_{x\in\cX}p_x|x\rangle\langle x|_X\ox \rho_B^x$
where the classical register $X$ is with Alice and the quantum system $B$ is with Bob. 
To encode her message in the system $C$, Alice uses an arbitrary encoding function $e:\cX\rightarrow \cC$ (which we assume to be surjective, i.e.~$\cC=e(\cX)$). 
The classical register $C$ is then defined as the one associated with the Hilbert space $\cH_C$ with orthonormal basis $\lbrace |c\rangle\rbrace_{c\in\cC}$. 
The encoding function $e$ gives rise to an encoding (quantum) map, which without loss of generality can be taken to be an isometric encoding map  $U_e\colon |x\rangle \mapsto |x\rangle\ox|e(x)\rangle$ where $|e(x)\rangle\in\lbrace |c\rangle\rbrace_{c\in\cC}$ for all $x\in\cX$. 
This results in the state 
\begin{align}\label{eq:dc-omega}
\omega_{XCB} \coloneqq U_e\,\rho_{XB}U_e^\dagger = \sumi_{x\in\cX} p_x |x\rangle\langle x|_X\ox |e(x)\rangle\langle e(x)|_C \ox \rho_B^x.
\end{align} 

Upon receiving the classical message $C$, Bob applies a measurement given by a POVM $\Lambda_c = \lbrace \Lambda_{x'\!,\,c}\rbrace_{x'\in\cX}$ to his state $\rho_B^x$, where the measurement is conditioned on the value of $c$ in the encoded register $C$. 
We label the corresponding random variable by $X'$. 
The final state of the classical registers $X$ and $X'$ is given by
\begin{align}\label{eq:dc-sigma}
\sigma_{XX'} \coloneqq \sumi_{x,x'\in\cX}p_x q_{x'|x} |xx'\rangle\langle xx'|_{XX'}\quad\text{where}\quad q_{x'|x}\coloneqq \tr\left(\Lambda_{x'\!,\,e(x)}\rho_B^x\right).
\end{align}
The POVM constitutes a quantum operation $\cD:CB\rightarrow X'$, that is,
\begin{align*}
\sigma_{XX'} = (\id_X\ox\cD)(\omega_{XCB}).
\end{align*}
Here, the CPTP map $\cD$ is defined as the one that, conditioned on the value $c$, applies the map 
\begin{align}\label{eq:dc-conditional-decoding}
\cD^c(\nu_B) = \sum_{x'\in\cX}\tr(\Lambda_{x'\!,\,c}\nu_B)|x'\rangle\langle x'|_{X'}\quad\text{for }\nu_B\in\cD(\cH_B),
\end{align}
implementing the POVM $\Lambda_c$ to the $B$ system, and then traces out $C$.
Note that $\cD^c$ is a special case of an entanglement-breaking channel \cite{HSR03} with Kraus operators given by 
\begin{align}\label{eq:dc-decoding-kraus-ops}
K_{x'j}^c \coloneqq |x'_{X'}\rangle\langle j_B|\sqrt{\Lambda_{x'\!,\,c}}
\end{align}
for an orthonormal basis $\lbrace |j_B\rangle\rbrace_{j=1}^{|B|}$ of $\cH_B$.

The final state $\sigma_{XX'}$ of the protocol should be close to the classically correlated target state $\vphi_{XX'}$ defined by
\begin{align}\label{eq:dc-phi}
\vphi_{XX'} \coloneqq  \sumi_{x\in\cX} p_x |xx\rangle\langle xx|_{XX'}.
\end{align}
As the figure of merit for a data compression protocol $(\rho_{XB},e,\cD)$ we choose the success probability $\psuc(\rho_{XB},e,\cD)$ of successfully decoding $X$:
\begin{align}\label{eq:dc-success-probability}
\psuc(\rho_{XB},e,\cD) \coloneqq \sumi_{x\in\cX} p_x q_{x|x}.
\end{align}

If Alice and Bob share $n$ copies of the c-q state $\rho_{XB}$, then the figure of merit for a data compression protocol $(\rho_{XB}^{\ox n},e_n,\cD_n)$ is given by
\begin{align}\label{eq:dc-figure-of-merit}
p_n\coloneqq \psuc(\rho_{XB}^{\ox n},e_n,\cD_n) = \sumi_{x^n\in\cX^n} p_{x^n} q_{x^n|x^n},
\end{align}
where $e_n\colon \cX^n\rightarrow \cC^n$ is the encoding function, and $p_{x^n}q_{x'^n|x^n}$ is the probability distribution of the classical state 
\begin{align*}
\sigma_{X^nX'^n}\coloneqq (\id_{X^n}\ox\cD_n)\left(U_{e_n}\rho_{XB}^{\ox n} U_{e_n}^\dagger\right).
\end{align*} 
Here, the CPTP map $\cD_n$ implements the POVM on the system $B^n$. The rate $m(\rho_{XB}^{\ox n},e_n,\cD_n)$ of the data compression protocol $(\rho_{XB}^{\ox n},e_n,\cD_n)$ is defined as
\begin{align}\label{eq:minimum-compression-length}
m(\rho_{XB}^{\ox n},e_n,\cD_n) \coloneqq \frac{1}{n} \log|C^n|.
\end{align}
A real number $m\geq 0$ is said to be an achievable rate for data compression with quantum side information, if there is a sequence $\lbrace (\rho_{XB}^{\ox n},e_n,\cD_n)\rbrace_{n\in\mathbb{N}}$ of protocols satisfying 
\begin{align*}
\limsup_{n\to \infty} m(\rho_{XB}^{\ox n},e_n,\cD_n) &= m, & \liminf_{n\to\infty} p_n &= 1.
\end{align*}

Data compression with quantum side information is the `dual' task  \cite{Ren10} to randomness extraction, with the optimal rate of the former also given by the conditional entropy $S(X|B)_\rho$. 
That is, a real number $m\geq 0$ is an achievable rate for data compression with quantum side information if and only if
\begin{align*}
m\geq  S(X|B)_\rho.
\end{align*}
This was proved by Winter \cite{Win99a} (see also Devetak and Winter \cite{DW03}). 
Renes and Renner \cite{RR12} derived one-shot bounds for data compression with quantum side information in terms of smooth entropies. 
Tomamichel \cite{Tom12} proved a strong converse theorem for data compression based on one-shot bounds in terms of smooth entropies.

\subsection{Strong converse theorem}\label{sec:dc-strong-converse}

\begin{lemma}\label{lem:dc-one-shot-converse}
Let $\rho_{XB}=\sum_{x\in\cX} p_x |x\rangle\langle x|_X\ox \rho_B^x$ be a c-q state with $\rho_B^x\in\cD(\cH_B)$ for all $x\in\cX$, and consider a data compression protocol $(\rho_{XB},e,\cD)$ as defined in \Cref{sec:dc-protocol}. 
We have the following bounds for the success probability $\psuc\equiv \psuc(\rho_{XB},e,\cD)$ defined in \eqref{eq:dc-success-probability}, for $\alpha\in(1/2,1)$ and $\beta=\alpha/(2\alpha-1)$:
\begin{align}
\log \psuc &\leq \frac{1-\alpha}{2\alpha} \left( \log|C| - S_\beta(XB)_\rho + S_\alpha(B)_\rho \right),\label{eq:dc-one-shot}\\
\log \psuc &\leq \frac{1-\alpha}{2\alpha} \left( \log|C| - \tS_\beta(X|B)_\rho\right).\label{eq:dc-one-shot-cond}
\end{align}
\end{lemma}

\begin{proof}
We first prove \eqref{eq:dc-one-shot}. 
Given an arbitrary POVM $\Lambda_c = \lbrace\Lambda_{x'\!,\,c}\rbrace_{x'\in\cX}$ for $c\in\cC$ and the corresponding conditional decoding map $\cD^c$ defined in \eqref{eq:dc-conditional-decoding} with Kraus operators $K^c_{x'\!,\,j}$ given by \eqref{eq:dc-decoding-kraus-ops}, it is straightforward to construct from this a Stinespring isometry $V_c\colon \cH_B\rightarrow \cH_{X'}\ox\cH_E$ for $\cD^c$:
\begin{align}\label{eq:stinespring-isometry}
V_c \coloneqq \sum_{x'\!,\,j} K^c_{x'\!,\,j}\ox |x'j_E\rangle = \sum_{x'\!,\,j} |x'_{X'}\rangle\langle j_B|\Lambda_{x'\!,\,c}^{1/2}\ox |x'j_E\rangle
\end{align}
where $\lbrace |x'j_E\rangle \rbrace_{x'\in\cX,j=1,\dots,|B|}$ is an orthonormal basis for the environment $\cH_E$ with $\dim\cH_E = |\cX||B|$, satisfying $\langle xj|yk\rangle_E = \delta_{xy}\delta_{jk}$. 
The isometry defined in \eqref{eq:stinespring-isometry} satisfies $V_c^\dagger V_c = \one_B$ and $\cD_c(\rho_B) = \tr_E(V_c \rho_B V_c^\dagger)$.

For every $x\in\cX$, let $|\psi^x_{BS}\rangle$ be a purification of $\rho_B^x$. Consider then the following purification of $\rho_{XB}$:
\begin{align*}
|\psi_{XBRS}\rangle \coloneqq \sum_{x\in\cX} \sqrt{p_x} |x_X\rangle |x_R\rangle |\psi^x_{BS}\rangle.
\end{align*}
Then the pure states obtained after encoding with $U_e$ and decoding with $\cD$ in the data compression protocol (cf.~\Cref{sec:dc-protocol}) are given by:
\begin{subequations}
\begin{align}
|\omega_{XCBRS}\rangle &= U_e |\psi_{XBRS}\rangle\nonumber\\
& =  \sum_{x\in\cX} \sqrt{p_x} |x_X\rangle |x_R\rangle |e(x)_C\rangle |\psi^x_{BS}\rangle,\label{eq:dc-omega-pure}\\
|\sigma_{XX'RSE}\rangle &= V |\omega_{XCBRS}\rangle\nonumber\\
& = \sum_{x,x'\in\cX}\sum_{j=1,\dots,|B|} \sqrt{p_x} |x_X\rangle |x_R\rangle |x'_{X'}\rangle \langle j_B|\Lambda_{x'\!,\,\,e(x)}^{1/2} |\psi^x_{BS}\rangle |x'j_E\rangle.\label{eq:dc-sigma-pure}
\end{align}
\end{subequations}
Here, $V$ denotes the Stinespring isometry of the overall decoding map $\cD$ that, conditional on the classical message $c\in\cC$, applies the decoding operation $\cD^c$. 
It can easily be checked that $|\omega_{XCBRS}\rangle$ and $|\sigma_{XX'RSE}\rangle$, as given by \eqref{eq:dc-omega-pure} and \eqref{eq:dc-sigma-pure} respectively, indeed purify $\omega_{XCB}$ and $\sigma_{XX'}$ as given by \eqref{eq:dc-omega} and \eqref{eq:dc-sigma}.

In the next step, we relate the final state $\sigma_{XX'}$ (resp.~its purification $|\sigma_{XX'RSE}\rangle$) of the data compression protocol to the target state $\vphi_{XX'}$ given by \eqref{eq:dc-phi}, thus obtaining a bound on the success probability $\psuc\equiv\psuc(\rho_{XB},e,\cD)$ defined in \eqref{eq:dc-success-probability}. 
To this end, we consider the following data compression protocol that allows perfect recovery of the register $X$, resulting in the target state $\vphi_{XX'}$: 
Take $C\cong X$ and consider the trivial encoding $e(x) = x$ for all $x\in\cX$, together with the trivial POVM $E_c = \lbrace \delta_{x'\!,\,c}\one_B\rbrace_{x'}$ which discards the quantum system $B$ and yields the message $x' = c = x$ with certainty. 
Following the same procedure as above, we can work out the pure states obtained after encoding and decoding using these particular choices:
\begin{subequations}
\begin{align}
|\bomega_{XCBRS}\rangle &= \sum_{x\in\cX} \sqrt{p_x} |x_X\rangle |x_R\rangle |x_C\rangle |\psi^x_{BS}\rangle,\\
|\bsigma_{XX'RSE}\rangle &= \sum_{x\in\cX}\sum_{j=1,\dots,|B|} \sqrt{p_x} |x_X\rangle |x_R\rangle |x_{X'}\rangle \langle j_B|\psi^x_{BS}\rangle |xj_E\rangle.
\end{align}
\end{subequations}
Note that $\bsigma_{XX'RSE}$ indeed purifies the target state $\vphi_{XX'}$ of the data compression protocol. Let us compute the inner product of the pure states $\sigma$ and $\bsigma$:
\begin{align*}
\langle \bsigma_{XX'RSE}|\sigma_{XX'RSE}\rangle &= \sum_{x,y,y'\in\cX}\sum_{j,k} \sqrt{p_x p_y} \langle xx|yy\rangle_{XR} \langle x|y'\rangle_{X'} \langle \psi^x_{BS}|j\rangle\langle k|_B \Lambda^{1/2}_{y'\!,\,e(y)} |\psi^x_{BS}\rangle \langle xj|y'k\rangle_E\\
&= \sum_{x\in\cX}\sum_j p_x \langle \psi^x_{BS}|j\rangle\langle j|_B \Lambda_{x,\,e(x)}^{1/2} | \psi^x_{BS}\rangle \\
&= \sum_{x\in\cX} p_x \langle \psi^x_{BS}| \Lambda_{x,\,e(x)}^{1/2} | \psi^x_{BS}\rangle \\
&\geq \sum_{x\in\cX} p_x \langle \psi^x_{BS}| \Lambda_{x,\,e(x)} | \psi^x_{BS}\rangle\\
&= \sum_{x\in\cX} p_x \tr\left(\Lambda_{x,\,e(x)}\rho_B^x\right)\\
&= \psuc,\numberthis\label{eq:dc-psucc-bound}
\end{align*}
where in the third equality we used the completeness of the basis $\lbrace |j_B\rangle\rbrace_{j=1}^{|B|}$, and in the inequality we used the fact that 
\begin{align*}
\sqrt{\Lambda_{x'\!,\,c}} \geq \Lambda_{x'\!,\,c}\quad \text{for all $x'\in\cX$ and $c\in\cC$},
\end{align*} 
since $0\leq \Lambda_{x'\!,c}\leq \one_B$ for all $x'\in\cX$ and $c\in\cC$. On the other hand, we have
\begin{align*}
\langle \bsigma_{XX'RSE}|\sigma_{XX'RSE}\rangle & = |\langle \bsigma_{XX'RSE}|\sigma_{XX'RSE}\rangle|\\
&= F\left(\bsigma_{XX'RSE}, \sigma_{XX'RSE}\right)\\
&\leq F\left(\bsigma_{XRS},\sigma_{XRS}\right)\\
&= F\left(\bsigma_{XRS},\omega_{XRS}\right),\numberthis\label{eq:dc-fidelity-bound}
\end{align*}
where we used the monotonicity of the fidelity under partial trace in the inequality, and the last line follows from the fact that the decoding $\cD$ does not affect the systems $X$, $R$, and $S$. 
Putting \eqref{eq:dc-psucc-bound} and \eqref{eq:dc-fidelity-bound} together, we obtain the following bound from \cref{eq:renyi-fidelity-entropy} of \Cref{prop:renyi-fidelity} for $\alpha\in(1/2,1)$ and $\beta = \alpha/(2\alpha-1)$:
\begin{align}\label{eq:dc-alpha-psucc-bound}
S_\alpha(XRS)_\omega \geq S_\beta(XRS)_{\bsigma} + \frac{2\alpha}{1-\alpha}\log \psuc.
\end{align}
By the duality property of the Rényi entropies (\Cref{prop:renyi-properties}(\ref{item:RE-duality})) and the subadditivity property (\Cref{lem:renyi-subadditivity}), we have
\begin{align}\label{eq:dc-subadditivity}
S_\alpha(XRS)_\omega = S_\alpha(CB)_\omega \leq \log |C| + S_\alpha(B)_\omega = \log |C| + S_\alpha(B)_\rho,
\end{align}
using the fact that $\omega_B=\rho_B$. 
Furthermore, due to the choice of the trivial encoding and decoding operations defined above and resulting in the state $\bsigma_{XX'} = \vphi_{XX'}$, we have 
\begin{align}\label{eq:ideal-encoder-renyi}
S_\beta(XRS)_{\bsigma} = S_\beta (X'E)_{\bsigma} = S_\beta(CB)_{\bomega} = S_\beta(XB)_\rho,
\end{align}
where we used the invariance of the Rényi entropies under isometries (\Cref{prop:renyi-properties}(\ref{item:RE-isom})) in the second equality, and the fact that $C\cong X$ is just a copy of the initial register $X$ in the last equality. 
Hence, combining \eqref{eq:dc-alpha-psucc-bound}, \eqref{eq:dc-subadditivity}, and \eqref{eq:ideal-encoder-renyi} we obtain
\begin{align*}
\log |C| + S_\alpha(B)_\rho \geq S_\beta(XB)_\rho + \frac{2\alpha}{1-\alpha}\log \psuc,
\end{align*}
which after re-arranging yields \eqref{eq:dc-one-shot}.

To prove \eqref{eq:dc-one-shot-cond}, consider the following chain of inequalities:
\begin{align*}
\frac{2\alpha}{1-\alpha}\log \psuc &\leq \frac{2\alpha}{1-\alpha}\log F(\sigma_{XX'},\varphi_{XX'})\\
&\leq \tS_\alpha(X|X')_\varphi - \tS_\beta(X|X')_\sigma\\
&= -\tS_\beta(X|X')_\sigma\\
&\leq -\tS_\beta(X|CB)_\omega\\
&\leq \log|C| - \tS_\beta(XC|B)_\omega\\
&= \log|C| - \tS_\beta(X|B)_\rho.
\end{align*}
The first inequality follows from \eqref{eq:dc-psucc-bound} and \eqref{eq:dc-fidelity-bound}, the second inequality uses \cref{eq:renyi-fidelity-conditional} of \Cref{prop:renyi-fidelity}, and the first equality uses the fact that $\tS_\alpha(X|X')_\varphi = 0$.
The third inequality is data processing with respect to the decoding map $\cD$, the fourth inequality follows from \cite[Prop.~8]{MDSFT13}, and the last equality uses the invariance of the Rényi conditional entropy with respect to the encoding isometry $U_e$.
\end{proof}

We then have the following strong converse theorem for data compression with quantum side information:

\begin{theorem}\label{thm:dc-strong-converse}
Let $\rho_{XB}$ be a c-q state and let $\lbrace (\rho_{XB}^{\ox n},e_n,\cD_n)\rbrace_{n\in\mathbb{N}}$ be a sequence of data compression protocols as defined in \Cref{sec:dc-protocol}, with figure of merit $p_n\equiv \psuc(\rho_{XB}^{\ox n},e_n,\cD_n)$ as defined in \eqref{eq:dc-figure-of-merit}. 
Then for all $n\in\mathbb{N}$ we have the following bounds on $p_n$ for $\alpha\in(1/2,1)$ and $\beta=\alpha/(2\alpha-1)$:
\begin{align}
p_n &\leq \exp\left\lbrace-n\kappa(\alpha)\left[S_\beta(XB)_\rho - S_\alpha(B)_\rho - m\right]\right\rbrace,\\
p_n &\leq \exp\left\lbrace-n\kappa(\alpha)\left[\tS_\beta(X|B)_\rho - m\right]\right\rbrace,
\end{align}
where $\kappa(\alpha)=(1-\alpha)/(2\alpha)$, and $m\equiv m(\rho_{XB}^{\ox n},e_n,\cD_n)$ denotes the compression rate as defined in \eqref{eq:minimum-compression-length}.
\end{theorem}

\begin{remark}
See \Cref{rem:bound-comparison} for a discussion of whether the two strong converse bounds of \Cref{thm:dc-strong-converse} are comparable.
\end{remark}

\section{Discussion and open questions}\label{sec:discussion}
For any information-theoretic task one can define {\em{achievable}} and {\em{strong converse rates}}. 
An achievable rate is a non-negative real number such that if one codes at a rate above it (for the case in which the optimal rate is a cost) then the error probability of the protocol vanishes asymptotically. 
The strong converse rate, in contrast, is a non-negative real number such that if one codes at a rate below it, then the error probability goes to one in the asymptotic limit.
Consequently, the fidelity between the initial and final target states of the protocol decays to zero asymptotically. 
The optimal rate of the protocol is said to satisfy the {\em{strong converse property}} if the largest strong converse rate coincides with the smallest achievable rate.
In this case the optimal rate provides a sharp rate threshold for the task. 
The exact exponent of the decaying fidelity (or success probability) for a given rate above zero and below the smallest achievable rate is called the {\em{strong converse exponent}}.\footnote{If the optimal rate is a \emph{gain} instead of a cost, one needs to swap `above' and `below' as well as `smallest' and `largest' in the above paragraph.}

In this paper, we used the Rényi entropy method to derive strong converse theorems which establish the strong converse property of the optimal rates for the following protocols: 
state redistribution (without and with feedback) (\Cref{thm:state-re-strong-converse} and \Cref{thm:state-re-fb-strong-converse}), coherent state merging (\Cref{thm:fqsw-strong-converse}), quantum state splitting (\Cref{thm:qss-strong-converse}), measurement compression with quantum side information (\Cref{thm:meas-comp-strong-converse}), randomness extraction against quantum side information (\Cref{thm:randex-strong-converse}), and data compression with quantum side information (\Cref{thm:dc-strong-converse}). 

To this end, we established certain Rényi generalizations of the optimal rates of these protocols as bounds on the strong converse exponents. 
We derived inequalities involving these Rényi generalizations, the most important of which relate Rényi entropic quantities to the fidelity between two quantum states. 
These inequalities played a crucial role in the proofs of our strong converse theorems. 

Determining whether any of these Rényi generalizations are indeed the exact strong converse exponents for the information-theoretic tasks in question is an interesting problem for future research.

\paragraph{Acknowledgements.}
We would like to thank Renato Renner and Dave Touchette for helpful discussions. 
The hospitality of the Banff International Research Station (BIRS) during the workshop `Beyond IID in Information Theory' (5-10 July 2016), where part of the present work was done, is gratefully acknowledged.

\appendix
\section{Remaining proofs of Section \ref{sec:renyi-properties}}\label{sec:renyi-properties-proofs}

\subsection{Lemma \ref{lem:renyi-quantities}}\label{sec:renyi-properties-proofs-sub}

\begin{proof}[Proof of \eqref{eq:dim-bound-mutual} in \Cref{lem:renyi-quantities}]
Consider that
\begin{align*}
\tI_{\alpha}\left(  A;CB\right)  _{\rho} &  =\min_{\tau_{CB}}%
\frac{\alpha}{\alpha-1}\log\left\Vert \left(  \rho_{A}^{\left(
1-\alpha\right)  /2\alpha}\otimes\tau_{CB}^{\left(  1-\alpha\right)  /2\alpha
}\right)  \rho_{ABC}\left(  \rho_{A}^{\left(  1-\alpha\right)  /2\alpha
}\otimes\tau_{CB}^{\left(  1-\alpha\right)  /2\alpha}\right)  \right\Vert
_{\alpha}\\
&  =\min_{\tau_{CB}}\frac{\alpha}{\alpha-1}\log\left\Vert \tau_{CB}^{\left(
1-\alpha\right)  /2\alpha}\widetilde{\rho}_{ABC}\tau_{CB}^{\left(
1-\alpha\right)  /2\alpha}\right\Vert _{\alpha}+\frac{\alpha}{\alpha-1}%
\log\text{Tr}\left\{  \rho_{A}^{\left(  1-\alpha\right)  /\alpha}\rho
_{AB}\right\}  \\
&  =-\tS_{\alpha}(  A|CB)  _{\widetilde{\rho}}%
+\frac{\alpha}{\alpha-1}\log\text{Tr}\left\{  \rho_{A}^{\left(
1-\alpha\right)  /\alpha}\rho_{AB}\right\}  ,
\end{align*}
where we have defined the density operator
\begin{align*}
\widetilde{\rho}_{ABC}\equiv\frac{1}{\text{Tr}\left\{  \rho_{A}^{\left(
1-\alpha\right)  /\alpha}\rho_{ABC}\right\}  }\rho_{A}^{\left(
1-\alpha\right)  /2\alpha}\rho_{ABC}\rho_{A}^{\left(  1-\alpha\right)
/2\alpha}
\end{align*}
and have observed that%
\begin{align*}
\text{Tr}\left\{  \rho_{A}^{\left(  1-\alpha\right)  /\alpha}\rho
_{ABC}\right\}  =\text{Tr}\left\{  \rho_{A}^{\left(  1-\alpha\right)
/\alpha}\rho_{AB}\right\}  .
\end{align*}
Furthermore,
\begin{align*}
\widetilde{\rho}_{AB}=\text{Tr}_{C}\left\{  \widetilde{\rho}%
_{ABC}\right\}  =\frac{1}{\text{Tr}\left\{  \rho_{A}^{\left(  1-\alpha
\right)  /\alpha}\rho_{AB}\right\}  }\rho_{A}^{\left(  1-\alpha\right)
/2\alpha}\rho_{AB}\rho_{A}^{\left(  1-\alpha\right)  /2\alpha}.
\end{align*}
We now apply the bound from \Cref{lem:renyi-quantities}, \cref{eq:dim-bound-conditional},
\begin{align*}
-\tS_{\alpha}(  A|CB)  _{\widetilde{\rho}}%
\leq-\tS_{\alpha}(  A|B)  _{\widetilde{\rho}}%
+2\log\vert C\vert,
\end{align*}
to see that
\begin{align*}
&  -\tS_{\alpha}(  A|CB)  _{\widetilde{\rho}}%
+\frac{\alpha}{\alpha-1}\log\text{Tr}\left\{  \rho_{A}^{\left(
1-\alpha\right)  /\alpha}\rho_{AB}\right\}  \\
& \qquad\qquad \leq-\tS_{\alpha}(  A|B)  _{\widetilde{\rho}}%
+\frac{\alpha}{\alpha-1}\log\text{Tr}\left\{  \rho_{A}^{\left(
1-\alpha\right)  /\alpha}\rho_{AB}\right\}  +2\log\vert C\vert \\
& \qquad\qquad =\tI_{\alpha}(  A;Bt)  _{\rho}+2\log\vert
C\vert, 
\end{align*}
which yields \eqref{eq:dim-bound-mutual}.\end{proof}

\begin{proof}[Proof of \eqref{eq:product-mutual} in \Cref{lem:renyi-quantities}]
From the data processing inequality (\Cref{prop:renyi-properties}(\ref{item:dpi})) we know that
\begin{equation}
\tI_{\alpha}\left(  A;BC\right)  _{\rho\ox\sigma}\geq\tI_{\alpha
}\left(  A;B\right)  _{\rho}.
\end{equation}
On the other hand, consider that%
\begin{align}
\tI_{\alpha}(  A;BC)  _{\rho\ox\sigma}  &  =\min_{\tau_{BC}%
}\tD_{\alpha}(  \rho_{AB}\otimes\sigma_{C}\Vert\rho
_{A}\otimes\tau_{BC}) \\
&  \leq\tD_{\alpha}(  \rho_{AB}\otimes\sigma_{C}\Vert\rho
_{A}\otimes\theta_{B}\otimes\sigma_{C}) \\
&  =\tD_{\alpha}(  \rho_{AB}\Vert\rho_{A}\otimes\theta
_{B})  .
\end{align}
Since the inequality holds for all $\theta_{B}$, we get that%
\begin{equation}
\tI_{\alpha}(  A;BC)  _{\rho\ox\sigma}\leq\tI_{\alpha
}(  A;B)  _{\rho},
\end{equation}
and we are done.
\end{proof}

\subsection{Proposition \ref{prop:renyi-fidelity}}

\begin{proof}[Proof of \eqref{eq:renyi-fidelity-mutual} in \Cref{prop:renyi-fidelity}]
For an arbitrary density operator $\tau_B$ and $\eps\in(0,1)$ define the state $\tau(\eps)_B\coloneqq (1-\eps)\tau_B + \eps\pi_B$. 
We then have the following chain of inequalities:
\begin{align*}
& \tD_\beta(\rho_{AB}\|\rho_A\ox \tau_B) - \tD_\alpha(\sigma_{AB}\|\sigma_A\ox \tau(\eps)_B) - \log (1-\eps)\\
&\quad = \tD_\beta(\rho_{AB}\| \rho_A\ox (1-\eps)\tau_B) - \tD_\alpha(\sigma_{AB}\|\sigma_A\ox \tau(\eps)_B)\\
&\quad \geq \tD_\beta(\rho_{AB}\|\rho_A \ox \tau(\eps)_B) - \tD_\alpha (\sigma_{AB}\|\sigma_A \ox \tau(\eps)_B)\\
&\quad = \frac{2\beta}{\beta-1} \log\left\| \rho_{AB}^{1/2} \rho_A^{(1-\beta)/2\beta} \ox \tau(\eps)_B^{(1-\beta)/2\beta} \right\|_{2\beta} - \frac{2\alpha}{\alpha-1}\log \left\| \sigma_A^{(1-\alpha)/2\alpha} \ox \tau(\eps)_B^{(1-\alpha)/2\alpha} \sigma_{AB}^{1/2}  \right\|_{2\alpha}\\
&\quad = \frac{2\alpha}{1-\alpha}\log\left[ \left\| \rho_{AB}^{1/2} \rho_A^{(1-\beta)/2\beta} \ox \tau(\eps)_B^{(1-\beta)/2\beta} \right\|_{2\beta}  \left\|  \sigma_A^{(1-\alpha)/2\alpha} \ox \tau(\eps)_B^{(1-\alpha)/2\alpha} \sigma_{AB}^{1/2} \right\|_{2\alpha}  \right]\\
&\quad \geq \frac{2\alpha}{1-\alpha} \log \left\| \rho_{AB}^{1/2} \left(\rho_A^{(1-\beta)/2\beta} \ox \tau(\eps)_B^{(1-\beta)/2\beta}\right) \left(\sigma_A^{(1-\alpha)/2\alpha} \ox \tau(\eps)_B^{(1-\alpha)/2\alpha}\right) \sigma_{AB}^{1/2}\right\|_1\\
&\quad = \frac{2\alpha}{1-\alpha} \log \left\| \rho_{AB}^{1/2} \left( \rho_A^{(1-\beta)/2\beta}\sigma_A^{(1-\alpha)/2\alpha} \ox \tau(\eps)_B^{(1-\beta)/2\beta} \tau(\eps)_B^{(1-\alpha)/2\alpha} \right) \sigma_{AB}^{1/2}\right\|_1\\
&\quad = \frac{2\alpha}{1-\alpha} \log \left\|\rho_{AB}^{1/2} \sigma_{AB}^{1/2}\right\|_1\\
&\quad = \frac{2\alpha}{1-\alpha} \log F(\rho_{AB},\sigma_{AB}).
\end{align*}
In the first equality and inequality we used \eqref{eq:c-relation} and \eqref{eq:domination-relation}. 
The following equalities follow from the definition of the sandwiched Rényi divergence (see \Cref{def:renyi-entropies}) and \eqref{eq:beta-alpha-relation}. 
In the second inequality we apply H\"older's inequality \eqref{eq:hoelder-inequality}. 
For the second-to-last equality we use the fact that $\rho_A=\sigma_A$ by assumption, and that $\tau(\eps)_B$ has full support for $\eps\in(0,1)$, such that $\tau(\eps)_B^{-1}\tau(\eps)_B = \one_B$.

We therefore have
\begin{align*}
\frac{2\alpha}{1-\alpha}\log F(\rho_{AB},\sigma_{AB}) &\leq \tD_\beta(\rho_{AB}\|\rho_A\ox \tau_B) - \tD_\alpha(\sigma_{AB}\|\sigma_A\ox \tau(\eps)_B) - \log (1-\eps)\\
&\leq \tD_\beta(\rho_{AB}\|\rho_A\ox \tau_B) - \min_{\omega_B\in\cD(\cH_B)} \tD_\alpha(\sigma_{AB}\|\sigma_A\ox \omega_B) - \log (1-\eps)\\
&= \tI_\beta(A;B)_\rho - \tI_\alpha(A;B)_\sigma - \log(1-\eps),
\end{align*}
where we chose $\tau_B$ as the optimizing state for $\tI_\beta(A;B)_\rho$ in the last step. 
Since this relation holds for all $\eps\in(0,1)$, we obtain the claim by taking the limit $\eps \searrow 0$.
\end{proof}

\begin{proof}[Proof of \eqref{eq:renyi-fidelity-cmi} in \Cref{prop:renyi-fidelity}]
We can rewrite definition \eqref{eq:renyi-cmi} in \Cref{sec:renyi-properties} as
\begin{align*}
\tI_{\beta}(  A;B|C)  _{\rho}=-\frac{2\alpha}{\alpha
-1}\log\left\Vert \rho_{BC}^{\left(  1-\beta\right)  /2\beta}\rho_{C}^{\left(
\beta-1\right)  /2\beta}\rho_{AC}^{\left(  1-\beta\right)  /2\beta}%
\rho_{ABC}^{1/2}\right\Vert _{2\beta}.
\end{align*}
Then consider
\begin{align*}
& \frac{1-\alpha}{2\alpha}\left[ \tI_{\beta}\left(  A;B|C\right)  _{\rho}-\tI_{\alpha
}(  A;B|C)  _{\sigma}\right]\\
&\qquad  =\log\left[  \left\Vert \sigma_{ABC}^{1/2}\rho
_{AC}^{\left(  1-\alpha\right)  /2\alpha}\rho_{C}^{\left(  \alpha-1\right)
/2\alpha}\rho_{BC}^{\left(  1-\alpha\right)  /2\alpha}\right\Vert _{2\alpha
}\left\Vert \rho_{BC}^{\left(  1-\beta\right)  /2\beta}\rho_{C}^{\left(
\beta-1\right)  /2\beta}\rho_{AC}^{\left(  1-\beta\right)  /2\beta}%
\rho_{ABC}^{1/2}\right\Vert _{2\beta}\right] \\
&\qquad  \geq\log\left[  \left\Vert \sigma_{ABC}^{1/2}%
\rho_{AC}^{\left(  1-\alpha\right)  /2\alpha}\rho_{C}^{\left(  \alpha
-1\right)  /2\alpha}\rho_{BC}^{\left(  1-\alpha\right)  /2\alpha}\rho
_{BC}^{\left(  1-\beta\right)  /2\beta}\rho_{C}^{\left(  \beta-1\right)
/2\beta}\rho_{AC}^{\left(  1-\beta\right)  /2\beta}\rho_{ABC}^{1/2}%
\right\Vert _{1}\right] \\
&\qquad  =\log\left\Vert \sigma_{ABC}^{1/2}\rho_{ABC}%
^{1/2}\right\Vert _{1}\\
&\qquad  =\log F(  \rho_{ABC},\sigma_{ABC}),
\end{align*}
which yields the claim.
\end{proof}

\section{Proofs of Section \ref{sec:feedback}: State redistribution with feedback}\label{sec:feedback-proofs}
The following lemma is used to prove the strong converse for state redistribution with feedback, \Cref{thm:state-re-fb-strong-converse}. 
Note that the proof closely follows that of the corresponding result in \cite{BCT14v2}.
\begin{lemma}\label{lem:state-re-fb-sc}
Consider the fidelity
\begin{align}
F\coloneqq F\left(\psi_{A'B'C'R}\ox \Phi^m_{T_A'T_B'}, (\cD_M\circ\cE_M\circ\dots\circ\cD_1\circ\cE_1) \left(\psi_{ABCR}\ox\Phi_{T_AT_B}^k\right)\right),\label{eq:state-re-fb-fidelity}
\end{align} 
where the encoding and decoding maps $\cE_i$ and $\cD_i$ for $i=1,\dots,M$ are given as in \Cref{sec:feedback}. 
For $\alpha\in (1/2,1)$ and $\beta=\alpha/(2\alpha-1)$, we have the following bounds on $F$:
\begin{align}
\log F &\leq \frac{1-\alpha}{2\alpha} \left(\log|T_A| - \log |T_A'| + \sumi_{i=1}^M \log|Q_i| + \sumi_{i=1}^{M-1} \log |Q_i'| - S_\beta(AB)_\psi + S_\alpha(B)_\psi\right),\label{eq:state-re-fb-q+e}\\
\log F &\leq \frac{1-\alpha}{2\alpha} \left(2 \sumi_{i=1}^M \log|Q_i| - \tS_\beta(R|B)_\psi + \tS_\alpha(R|AB)_\psi\right),\label{eq:state-re-fb-q}\\
\log F &\leq \frac{1-\alpha}{2\alpha} \left(2 \sumi_{i=1}^M \log|Q_i| - \tI_\alpha(R;AB)_\psi + \tI_\beta(R;B)_\psi\right).\label{eq:state-re-fb-q-alt}
\end{align}
\end{lemma}

\begin{proof}
We first prove \eqref{eq:state-re-fb-q}. For $\alpha\in(1/2,1)$ and $\beta=\alpha/(2\alpha-1)$, we can bound the fidelity $F$ (defined in \eqref{eq:state-re-fb-fidelity}) from above by
\begin{align}
\frac{2\alpha}{1-\alpha}\log F &\leq \frac{2\alpha}{1-\alpha} \log F\left(\psi_{A'B'R}\ox \pi_{T_B'}^m, \sigma^M_{A'B'RT_B'}\right)\notag\\
&\leq \tS_\alpha(R|A'B'T_B')_{\psi\ox\pi^m} - \tS_\beta(R|A'B'T_B')_{\sigma^M}\notag\\
&= \tS_\alpha(R|AB)_{\psi} - \tS_\beta(R|A'B'T_B')_{\sigma^M},\label{eq:fb-fidelity-bound}
\end{align}
where we used the monotonicity of the fidelity under partial trace in the first inequality, \cref{eq:renyi-fidelity-conditional} of \Cref{prop:renyi-fidelity} in the second inequality, and \cref{eq:product-conditional} of \Cref{lem:renyi-quantities} together with the fact that $\psi_{A'B'R}=\psi_{ABR}$ in the equality. 
Consider then the following chain of inequalities for the second term on the right-hand side of \eqref{eq:fb-fidelity-bound} (see \Cref{fig:state-re-fb-protocol}):
\begin{align*}
-\tS_\beta(R|A'B'T_B')_{\sigma^M} &\leq - \tS_\beta(R|Q_M B_{M-1})_{\omega^M}\\
&\leq - \tS_\beta(R|B_{M-1} )_{\omega^M} + 2\log|Q_M|\\
&= - \tS_\beta(R| B_{M-1})_{\sigma^{M-1}} + 2\log|Q_M|\\
&\leq - \tS_{\beta}(R|Q_{M-1}'B_{M-1})_{\sigma^{M-1}} + 2\log |Q_M|\\
&\leq - \tS_\beta (R|Q_{M-1} B_{M-2})_{\omega^{M-1}} + 2\log |Q_M|\\
&\;\;\vdots\\
&\leq -\tS_\beta(R|BT_B)_{\omega^1} + 2\sum_{i=1}^M \log |Q_i|\\
&= -\tS_\beta(R|BT_B)_{\psi\ox\pi^k} + 2\sum_{i=1}^M \log |Q_i|\\
&= -\tS_\beta(R|B)_\psi + 2\sum_{i=1}^M \log |Q_i|.\numberthis\label{eq:fb-chain-ineq}
\end{align*}
In the first inequality we used data processing with respect to the decoding map $\cD_M$ (\Cref{prop:renyi-properties}(\ref{item:dpi})). 
The second inequality follows from the dimension bound for the Rényi conditional entropy (\cref{eq:dim-bound-conditional} of \Cref{lem:renyi-quantities}). 
In the first equality we used the fact that the system $B_{M-1}$ is not affected by the encoding $\cE_M$. 
The third inequality is data processing for the Rényi conditional entropy with respect to the partial trace over $Q_{M-1}'$. 
We then iteratively apply these steps until we reach the last inequality. 
The subsequent equality follows from the fact that the encoding $\cE_1$ does not act on the systems $B$ and $T_B$. 
In the last step we used \cref{eq:product-conditional} of \Cref{lem:renyi-quantities}. 
Combining \eqref{eq:fb-fidelity-bound} and \eqref{eq:fb-chain-ineq} now yields \eqref{eq:state-re-fb-q}.
The proof of the bound in \eqref{eq:state-re-fb-q-alt} follows in a similar manner, and we therefore omit it.

To prove \eqref{eq:state-re-fb-q+e}, we consider Stinespring isometries $U_{\cE_i}$ and $U_{\cD_i}$ of the encoding and decoding maps $\cE_i$ and $\cD_i$ with environments $E_i$ and $D_i$, respectively. 
Moreover, in the following calculations we denote by $\omega^i$ and $\sigma^i$ the \emph{pure states} obtained from applying the isometries $U_{\cE_i}$ and $U_{\cD_i}$ to the initial state $\psi\ox\Phi^k$, respectively. 
The final state of the protocol is then the pure state 
\begin{align*} 
|\sigma^M_{A'B'C'RT_A'T_B'E_1\dots E_MD_1\dots D_M}\rangle = (U_{\cD_M}U_{\cE_M}\dots U_{\cD_1}U_{\cE_1}\ox \one_R)(|\psi_{ABCR}\rangle\ox|\Phi^k_{T_AT_B}\rangle).
\end{align*} 
By Uhlmann's theorem there exists a pure state $\chi_{E_1\dots E_M D_1\dots D_M}$ such that the following holds for $\alpha\in(1/2,1)$ and $\beta=\alpha/(2\alpha-1)$: 
\begin{align}
\frac{2\alpha}{1-\alpha}\log F &= \frac{2\alpha}{1-\alpha}\log F\left(\sigma^M_{A'B'C'RT_A'T_B'E_1\dots E_M D_1\dots D_M}, \psi_{A'B'C'R}\ox\Phi^m_{T_A'T_B'}\ox\chi_{E_1\dots E_M D_1\dots D_M}\right)\notag\\
&\leq \frac{2\alpha}{1-\alpha}\log F\left(\sigma^M_{A'B'T_B'D_1\dots D_M}, \psi_{A'B'}\ox\pi^m_{T_B'}\ox\chi_{D_1\dots D_M}\right)\notag\\
&\leq S_\alpha(A'B'T_B'D_1\dots D_M)_{\sigma^M} - S_\beta(A'B'T_B'D_1\dots D_M)_{\psi\ox\pi^m\ox\chi}\notag\\
&\leq S_\alpha(A'B'T_B'D_1\dots D_M)_{\sigma^M} - S_\beta(AB)_\psi - \log|T_B'|, \label{eq:fb-fidelity-bound-ent}
\end{align}
where the first inequality follows from the monotonicity of the fidelity under partial trace, the second inequality follows from \cref{eq:renyi-fidelity-entropy} of \Cref{prop:renyi-fidelity}, and the third inequality follows from \Cref{prop:renyi-properties}(\ref{item:RE-dim-bound}) and (\ref{item:RE-add}).
For the first term of the right-hand side of \eqref{eq:fb-fidelity-bound-ent}, consider the following steps:
\begin{align*}
S_\alpha(A'B'T_B'D_1\dots D_M)_{\sigma^M} &= S_\alpha(Q_M B_{M-1}D_1\dots D_{M-1})_{\omega^M}\\
&\leq S_\alpha(B_{M-1}D_1\dots D_{M-1})_{\omega^M} + \log|Q_M|\\
&= S_\alpha(RQ_MC'T_A'E_1\dots E_M)_{\omega^M} + \log|Q_M|\\
&= S_\alpha(R Q_{M-1}' A_{M-1} E_1\dots E_{M-1})_{\sigma^{M-1}} + \log|Q_M|\\
&\leq S_\alpha(R A_{M-1} E_1\dots E_{M-1})_{\sigma^{M-1}} + \log|Q_M| + \log|Q_{M-1}'|\\
&= S_\alpha(Q_{M-1}'B_{M-1} D_1\dots D_{M-1})_{\sigma^{M-1}} + \log|Q_M| + \log|Q_{M-1}'|\\
&= S_\alpha(Q_{M-1}B_{M-2} D_1\dots D_{M-2})_{\omega^{M-1}} + \log|Q_M| + \log|Q_{M-1}'|\\
&\;\;\vdots\\
&\leq S_\alpha(BT_B)_{\omega^1}+ \sum_{i=1}^M\log|Q_i| + \sum_{i=1}^{M-1}\log|Q_i'|\\
&= S_\alpha(BT_B)_{\psi\ox\pi^k} + \sum_{i=1}^M\log|Q_i| + \sum_{i=1}^{M-1}\log|Q_i'|\\
&= S_\alpha(B)_\psi + \log|T_B| + \sum_{i=1}^M\log|Q_i| + \sum_{i=1}^{M-1}\log|Q_i'|.\numberthis\label{eq:fb-ent-chain-ineq}
\end{align*}
In the first equality we used invariance of the Rényi entropy under the isometry $U_{\cD_M}$ (\Cref{prop:renyi-properties}(\ref{item:RE-isom})).
In the first inequality we used subadditivity (\Cref{lem:renyi-subadditivity}), and in the second equality we used the duality of the Rényi entropy since $|\omega^{M}\rangle$ is a pure state. 
The third equality follows from the invariance of the Rényi entropy under $U_{\cE_M}$. 
We then follow the same steps iteratively, passing from $\omega^M$ to $\sigma^{M-1}$ and $\omega^{M-1}$ and so on, until we reach $\omega^1_{BT_B}=\psi_B\ox\pi_{T_B}^k$. 
Substituting \eqref{eq:fb-ent-chain-ineq} in \eqref{eq:fb-fidelity-bound-ent} then yields \eqref{eq:state-re-fb-q+e}, and we are done.
\end{proof}

\printbibliography[title={References},heading=bibintoc]

\end{document}